\documentclass[11pt]{article}
\usepackage{amsmath, amsthm, amsfonts, amssymb}
\usepackage{arydshln}
\usepackage{authblk}

\usepackage[
backend=biber,
style=alphabetic,
sorting=nyt,
giveninits=true,
maxbibnames=10
]{biblatex}

\usepackage{cancel}
\usepackage{comment}
\usepackage{enumitem}
\setlist[itemize]{label=--}

\usepackage[colorlinks=true]{hyperref}
\hypersetup{
    linkcolor=blue,
    citecolor=blue,
    urlcolor=blue
}

\usepackage{mathtools}

\usepackage{xcolor}

\usepackage[colorinlistoftodos]{todonotes}

\newcommand{\F}[0]{\mathbb{F}}
\newcommand{\Fq}[0]{\mathbb{F}_q}
\newcommand{\N}[0]{\mathbb{N}}

\newcommand{\es}[2]{e_{#1,\, #2}}

\DeclarePairedDelimiter{\fl}{\lfloor}{\rfloor}
\DeclarePairedDelimiter{\cl}{\lceil}{\rceil}

\theoremstyle{plain}
\newtheorem{theo}{Theorem}[section]
\newtheorem{prop}[theo]{Proposition}
\newtheorem{lem}[theo]{Lemma}
\newtheorem{cor}[theo]{Corollary}
\newtheorem{notat}{Notation}[section]

\theoremstyle{definition}
\newtheorem{defin}{Definition}[section]
\newtheorem{ex}{Example}[section]

\theoremstyle{remark}

\newtheorem{prob}{Problem}

\title{Symmetric Models for Syndrome Decoding}
\author[1]{Elisa Gorla}
\author[1]{Simone Trebiani}
\date{}
\affil[1]{Institut de mathématiques, Université de Neuchâtel}

\begin{document}

\maketitle

\noindent\textbf{Abstract.} This paper introduces a new polynomial model for the exact variant of the Syndrome Decoding Problem (SDP) in the binary case. The model is based on elementary symmetric polynomials. We estimate the computational complexity of solving the corresponding polynomial system by establishing bounds on the degree of regularity and on the solving degree of the ideal associated to the model. The complexity estimate is lower than for previous polynomial models.
We also provide a variant of the model whose complexity depends directly on the specific instance of the SDP and is lower than for the first model. Finally, we discuss how to apply our approach to solve other variants of the SDP.\\


\section*{Introduction}

Over the past few years, there has been growing interest in new branches of cryptography, driven by the search for problems capable of resisting attacks from the emerging quantum computing. Several proposed schemes are based on problems from coding theory, thus giving rise to \textit{code-based cryptography}. For example, three of them (BIKE, Classic McEliece, and HQC) were included by the \textit{National Institute of Standards and Technology} in the list of candidates in the fourth and final round of the Post-Quantum Cryptography Standardization Process\footnote{\url{https://csrc.nist.gov/projects/post-quantum-cryptography/post-quantum-cryptography-standardization/round-4-submissions}}. HQC was one of two key-encapsulation mechanisms selected for standardization\footnote{\url{https://csrc.nist.gov/Projects/post-quantum-cryptography/post-quantum-cryptography-standardization/selected-algorithms}}. Two more code-based schemes (CROSS and LESS) were included in the list of second round candidates in the Additional Digital Signature Schemes PQC Standardization Process\footnote{\url{https://csrc.nist.gov/Projects/pqc-dig-sig/round-2-additional-signatures}}, although they were not selected for the third round~\cite{NISTIR8610}. 

For this reason, assessing the hardness of the underlying mathematical problems becomes of central importance. One can do so by estimating the computational complexity of the best known algorithms or by exploring new attack strategies. One such strategy is to model the problem as a polynomial system and then to use Gr\"obner bases techniques to solve it. For the \textit{Syndrome Decoding Problem} (SDP), a first step in this direction was taken in~\cite{MPSmodels}, where a polynomial model for an instance of the SDP was proposed. A further step was taken in~\cite{caminata2025quadraticmodelingssyndromedecoding}, where the original model was improved and an estimate of its computational complexity was given. In the same work, the authors also provided a polynomial model for the SDP over larger fields.\\

\noindent\textit{Our contributions \& Structure of the paper.} In Section~\ref{sec:Preliminaries} we recall some preliminaries on the Syndrome Decoding Problem and the Polynomial System Solving Problem and give a brief overview of the previous polynomial models for the SDP. In Section~\ref{secFirstModel} we propose a new polynomial model based on elementary symmetric polynomials for the Exact Syndrome Decoding Problem over the binary field, and we investigate its computational complexity, obtaining a bound that is lower than that of~\cite{caminata2025quadraticmodelingssyndromedecoding} (Theorem~\ref{cordsolS}). In Section~\ref{secComplexityESDP} we then propose a variant of the model, for which the complexity of solving the polynomial system depends more clearly on the specific instance of the SDP and is even lower than the previous result (Corollary~\ref{cordsolS'}).
In Section~\ref{sezL}, we discuss a linear algebra problem that arises from the variant model. Finally, in the last section we discuss how to use the ideas of the previous sections to solve related problems, such as the Bounded SDP (Section~\ref{secBSDP}), the non-binary version of the SDP (Section~\ref{secSizeq>2}), and the Regular SDP (Section~\ref{sezRSDP}).

\section{Preliminaries}\label{sec:Preliminaries}

\subsection{Notation}
Let $\Fq$ be a finite field with $q \geq 2$ elements and $R_n$ the polynomial ring $\Fq[x_1,\dots,x_n]$ with $n \geq 1$ variables and coefficients in $\Fq$. $R_{n,d}$ denotes the vector space of homogeneous polynomials of $R_n$ of degree $d$. $(f_1,\dots,f_j)$ denotes the ideal generated by the polynomials $f_1,\dots,f_j$.

Let $f\in R_n$ be a polynomial and let $f^{\operatorname{top}}$ denote the homogeneous component of highest degree of $f$. If $S=\{f_1,\dots,f_j\}$ is a set of polynomials, $S^{\operatorname{top}}=\{f_1^{\operatorname{top}},\dots,f_j^{\operatorname{top}}\}$ is the set of the homogeneous components of highest degree of the elements of $S$.

Let $[n]$ denote the set $\{1,\dots,n\}$ and $S_n$ its group of permutations.

For $\alpha \subseteq [n]$, let $x_\alpha= x_{\alpha_1}\cdots x_{\alpha_i}$ be the squarefree monic monomial in $R_n$ that is the product of the variables indexed by the elements of $\alpha$.

Given a set $A$, the cardinality of $A$ is denoted by $|A|$. So, for any given subset $\alpha\subseteq [n]$ as above, we have $\operatorname{deg}(x_\alpha)=\operatorname{deg}(x_{\alpha_1}\cdots x_{\alpha_i})=\lvert\alpha\rvert$.
 
Given $n,m\geq 1$, let $\F_q^{m\times n}$ be the set of matrices with $m$ rows, $n$ columns and coefficients in $\F_q$. We denote the determinant and the rank of a matrix $A$ by $\operatorname{det}(A)$ and $\operatorname{rk}(A)$, respectively.

\subsection{The Syndrome Decoding Problem}
Before defining the Syndrome Decoding Problem, we recall the usual notation for linear codes. Given integers $n\geq 2$ and $k\in [n]$, a $[n,k]$-\textit{linear code} $\mathbf{C}$ is a $k$-dimensional subspace of $\Fq^n$, where $n$ is its \textit{length} and $k$ its \textit{dimension}. An element of $\mathbf{C}$ is called a \textit{codeword}, and we define the \textit{Hamming weight} of a codeword $c\in\mathbf{C}$, denoted by $\operatorname{wt}(c)$, as the number of its non-zero entries. The \textit{minimum distance} of $\mathbf{C}$, denoted by $d(\mathbf{C})$, is the Hamming weight of the smallest non-zero codeword of $\mathbf{C}$.

Given a $[n, k]$-linear code $\mathbf{C}$, a full rank matrix $\mathbf{H}\in\F_q^{(n-k)\times n}$ is a \textit{parity check matrix} of $\mathbf{C}$ if $\mathbf{H}\cdot c^\top=0$ holds for every codeword $c$ of $\mathbf{C}$. Given a parity check matrix $\mathbf{H}$ of $\mathbf{C}$, the \textit{syndrome} of a vector $e\in\F_q^n$ is $s=\mathbf{H}\cdot e^\top\in\F_q^{n-k}$.

We can now describe the Syndrome Decoding Problem (SDP) in two variants: Given $n\geq 2$ and $k\in [n]$, a parity check matrix \mbox{$\mathbf{H}\in\F_q^{(n-k)\times n}$} of a $\left[n, k\right]$-linear code $\mathbf{C}$, let $s$ be a vector in $\F_q^{n-k}$ and $t$ be an integer such that $0\leq t\leq n$.

\begin{prob}[(B)SDP: (Bounded) Syndrome Decoding Problem]
    Find a vector $e\in\F_q^n$ such that $\mathbf{H}\cdot e^\top=s^\top$ and $\operatorname{wt}(e)\leq t$.
\end{prob}

\begin{prob}[ESDP: Exact weight Syndrome Decoding Problem]
    Find a vector $e\in\F_q^n$ such that $\mathbf{H}\cdot e^\top=s^\top$ and $\operatorname{wt}(e)=t$.
\end{prob}

A triple $(\mathbf{H},s,t)$ as above is an \textit{instance} of the SDP.\\

The SDP is known to be NP-complete~\cite{1055873} and believed to be hard on average. For this reason, some code-based cryptosystems base their security on the hardness of the SDP, most notably the McEliece cryptosystem~\cite{mceliece1978public} and its later variants, such as the one~\cite{ClassicMcEliece2022} proposed at the NIST competition~\cite{Alagic2025NISTIR8545}. In the McEliece settings, the public key is a seemingly random code $\mathbf{C'}$, which is a scrambled version of a $t$-error correcting code $\mathbf{C}$, which is the private key. The encryption of a message $m$ is its encoding through $\mathbf{C'}$, with the addition of some noise $e$ of weight $t$. The original message can be recovered thanks to the error correction properties of $\mathbf{C}$. An equivalent version based directly on the parity check matrix, known as the Niederreiter scheme~\cite{1571980076566605568}, can be interpreted by an attacker as an instance of the SDP.

The maximum weight of a uniquely correctable error $e$ depends on the properties of $\mathbf{C}$. More precisely, if the target weight $t$ is chosen such that $t\leq\fl*{\frac{d(\mathbf{C})-1}{2}}$, the uniqueness of the solution and, therefore, the recovery of the message is guaranteed. An SDP-based cryptographic scheme is generally constructed so that the weight $t$ of the inserted error is maximized, while still guaranteeing unique decoding.

\subsection{The Polynomial System Solving Problem}\label{sec:posso}
In this paper, we devise a polynomial system whose solutions coincide with the solutions of the Syndrome Decoding Problem. By doing so, we reduce the SDP to the problem of solving a polynomial system. The related problem is therefore:
\begin{prob}[PoSSo: Polynomial System Solving]
    Given a set of $r$ polynomials $F=\{f_1,\dots,f_r\}\subseteq R_n$, find a solution of the system $F=0$, that is, find a vector $y=(y_1,\dots,y_n)\in\Fq^n$ such that $f_1(y)=\dots=f_r(y)=0$.
\end{prob}
PoSSo was proven to be NP-complete in the case $q=2$ in~\cite[28]{FRAENKEL197915}. In addition, it is widely believed to be an average-case difficult problem. This is witnessed by the fact that a whole branch of post-quantum cryptography, called \textit{multivariate cryptography}, relies on its hardness. \\

The generic approach for solving the PoSSo problem is via \textit{Gr\"obner bases}. Gr\"obner bases were introduced by Buchberger in~\cite{BUCHBERGER2006475} and may be regarded a ``nice" system of generators of the ideal $(F)$. Over a finite field and for a system that has finitely many solutions over any field extension, the solutions can be computed in polynomial time from the reduced lexicographic Gr\"obner basis of the system. We refer the interested reader to~\cite{kreuzerrobbiano2000} for a mathematical treatment of Gr\"obner bases and to~\cite{Caminata_2021} for a discussion on how to use them in cryptography. 
%The algorithm is based on a generalized Shape Lemma and the factorization of some univariate polynomials: this can be done in polynomial time (for example, with Berlekamp's algorithm~\cite{berlekamp, Geddes1992AlgorithmsCA}), once one has the Gr\"obner basis. The issue is that obtaining it can be computationally demanding, especially for the lexicographic ones. In fact, the most convenient strategy is to compute a Gr\"obner basis with respect to another term order, the graded reverse lexicographic order (\texttt{degrevlex}), and then apply an algorithm, such as the FGLM algorithm~\cite{FAUGERE1993329}, to convert it into the reduced lexographic one. The complexity of FGLM is polynomial in the number of variables and in the number of solutions of the associated system, and since this is also true for the algorithm for computing the solutions using the Gr\"obner basis, 
In~\cite{Caminata_2021} it is argued that, for systems that are currently relevant for cryptography, the complexity of computing their solutions using Gr\"obner bases is dominated by the complexity of computing a degree-reverse lexicographic (\texttt{degrevlex}) Gr\"obner basis. This will therefore be the focus of the rest of our discussion.

There are many algorithms for the computation of Gr\"obner bases, such as F4~\cite{FAUGERE199961}, F5~\cite{F5}, and XL~\cite{XL} and its variants~\cite{buchmann_et_al:DagSemProc.09031.10, 10.1007/978-3-540-88403-3_14}, which all rely on Gaussian elimination on the \textit{Macaulay matrices} associated with the polynomial system $F=\{f_1,\dots,f_r\}$. The Macaulay matrix $M_{\leq d}$ of degree $d$ of $F$ is defined as follows: It has columns indexed by the monic monomials $x_\beta$ of degree $\leq d$ (ordered according to a specific term order), rows indexed by the polynomials $x_{\alpha}f_i$, where $x_{\alpha}$ is a monic monomial such that $\operatorname{deg}(x_{\alpha}f_i)\leq d$, and its entry $((\alpha,i),\beta)$ is the coefficient of the monomial $x_\beta$ in the polynomial $x_{\alpha}f_i$. The least $d$ such that Gaussian elimination on the Macaulay matrix $M_{\leq d}$ produces a \texttt{degrevlex} Gr\"obner basis of $F$ is the \textit{solving degree} of $F$, denoted by $d_{\operatorname{sol}}(F)$. Therefore, up to constant factors, the complexity of computing a \texttt{degrevlex} Gr\"obner basis is asymptotically dominated by the complexity of performing Gaussian elimination~\cite[111-118]{matrixcomput} on the Macaulay matrix of degree $d_{\operatorname{sol}}(F)$. Since the matrix has $\sum\limits_{i=1}^r\binom{n+d_{\operatorname{sol}}(F)-\operatorname{deg}(f_i)}{d_{\operatorname{sol}}(F)-\operatorname{deg}(f_i)}$ rows and $\binom{n+d_{\operatorname{sol}}(F)}{d_{\operatorname{sol}}(F)}$ columns, the complexity is bounded from above by
\begin{equation*}
    O\left(\binom{n+d_{\operatorname{sol}}(F)}{d_{\operatorname{sol}}(F)}^\omega\,\right)
\end{equation*}
for sufficiently large $n, d_{\operatorname{sol}}(F)$, where $2\leq\omega<3$ is the \textit{matrix multiplication exponent}~\cite[313-316]{vonzurGathen_Gerhard_2013}.

The algorithms mentioned above employ their own strategies to reduce the number of operations required to compute a Gr\"obner basis, but we consider this rough estimate to keep the discussion as general as possible.\\

The issue is that the solving degree of a system $F$ may be difficult to compute or estimate: for this reason, it has been related to several algebraic invariants (see e.g.~\cite{CAMINATA2023322}), such as the \textit{degree of regularity}, introduced by Bardet, Faugère, and Salvy in~\cite{Faugre2004OnTC}. The degree of regularity of $F$ is defined as
\begin{equation*}
    d_{\operatorname{reg}}(F)=\min\{d\in\N\mid (F^{\operatorname{top}})\cap R_{n,d}=R_{n,d}\}.
\end{equation*}
This invariant is often easier to compute than the solving degree. The next result by Salizzoni relates the degree of regularity and the solving degree. It allows us to focus on the degree of regularity in the sequel.
\begin{theo}[{\cite[Theorem 1]{salizzoni2023upperboundsolvingdegree}}]\label{thm:salizz}
    Let $F=\{f_1,\dots,f_r\}$ be a family of polynomials and let $\sigma$ be a degree-compatible term order. Then,
    \begin{equation*}
        d_{\operatorname{sol}}(F)\leq \operatorname{max}\{d_{\operatorname{reg}}(F)+1,\operatorname{deg}(f_1),\dots,\operatorname{deg}(f_r)\}.
    \end{equation*}    
\end{theo}
Since \texttt{degrevlex} is degree-compatible, we can use this result to estimate the solving degree, and thus the complexity of solving polynomial systems. This will allow us to compare different polynomial models. 
%Other estimations of the solving degree, such as the ones given in~\cite{salizzoni2023upperboundsolvingdegree, cryptoeprint:2019/903,tenti2019sufficiently}, are not suitable for this work, since the requirement on polynomial degrees will not always be matched. 

\subsection{Previous polynomial models for SDP}
There are several approaches to solve the SDP, one of the most prominent being \textit{Information Set Decoding}. The first algorithm was proposed by Prange~\cite{Prange1962TheUO}, followed by several subsequent improvements and variants (for more details, see e.g.~\cite{weger2024surveycodebasedcryptography}). 

Since the focus of this work is to provide a polynomial model to solve the SDP, we focus on known polynomial models. The first version by Meneghetti, Pellegrini, and Sala in~\cite{MPSmodels} was later improved by Caminata, Cartor, Meneghetti, Mora, and Pellegrini in~\cite{caminata2025quadraticmodelingssyndromedecoding}. Given an instance $(\mathbf{H},s,t)$ of the SDP over $\F_2$, \cite{MPSmodels} uses two different types of encoding variables: the variables $x_i$ are the entries of a solution vector $x$ of the SDP, while the variables $y_{i,j}$ are the digits of the binary expansion of the Hamming weight of $x$. The model involves the following polynomial equations:
\begin{itemize}
    \item the \textit{parity check encoding}, given by the linear equations $\mathbf{H}\cdot x^\top=s$,
    \item the \textit{Hamming weight computation encoding}, which forces the entries of the vector $y_i=(y_{i,1},\dots,y_{i,\fl*{\operatorname{log}_2 n}+1})$ to be the digits of the binary expansion of the Hamming weight of the truncated vector $(x_1,\dots,x_i)$, for $i=1,\dots,n$,
    \item the \textit{weight constraint encoding}, which enforces the constraint $\operatorname{wt}(x)\leq t$ or $\operatorname{wt}(x)=t$; 
    \item the \textit{field equations} $x_i^2=x_i$ and $y_{i,j}^2=y_{i,j}$.
\end{itemize}
The authors of~\cite{caminata2025quadraticmodelingssyndromedecoding} improve the model by lowering the degrees of the equations involved in the Hamming weight computation encoding and in the weight constraint encoding. They are able to produce a system whose equations have degree at most two. Moreover, after some linearization steps, the resulting system involves fewer variables and fewer equations than the original.

They also study the complexity of their model applied to the ESDP over $\F_2$. Combining the result of~\cite[Remark 8]{caminata2025quadraticmodelingssyndromedecoding} on the degree of regularity of the polynomial system $S_{\operatorname{QM}}=0$ %for some instance $(\mathbf{H},s,t)$ 
and Theorem~\ref{thm:salizz}, one obtains the following result:
\begin{theo}\label{teoSQM}
    The solving degree of $S_{\operatorname{QM}}$ is bounded from above by
    \begin{equation*}
        d_{\operatorname{sol}}(S_{\operatorname{QM}})\leq\cl*{\frac{n-1}{2}}+\fl*{\operatorname{log}_2t}\cl*{\frac{n-2}{2}}+2.
    \end{equation*}
\end{theo}
Moreover, the authors of~\cite{caminata2025quadraticmodelingssyndromedecoding} provide a variant of their system for the SDP that is valid for any $q\geq 2$. We briefly discuss it in Section~\ref{secSizeq>2}. The next section focuses on the binary case, which is the most relevant from a computational and cryptographic point of view. For example, it is the parameter choice in Classic McEliece~\cite{ClassicMcEliece2022}.

\section{New models for the ESDP via elementary symmetric polynomials}\label{secCore}

In this section, we propose two models for the ESDP that use elementary symmetric polynomials. Our goal is to produce a polynomial system that can be solved with lower complexity than the one in Theorem~\ref{teoSQM}. We focus on the binary case and let $q=2$ and $R_n=\F_2[x_1,\dots,x_n]$.

\subsection{A first model for the ESDP}\label{secFirstModel}

We start by recalling the definition of elementary symmetric polynomial. A polynomial in $R_n$ is \textit{symmetric} if it is invariant under any permutation of the variables $x_1,\ldots,x_n$, while an ideal is \textit{symmetric} if it contains the polynomials obtained from any polynomial in the ideal by permuting the variables. 
\begin{defin}
    Let $n\geq 1$ and $j\in[n]$. The \textit{elementary symmetric polynomial} of degree $j$ in the variables $x_1,\dots,x_n$ is
    \begin{equation*}
        \es{n}{j}(x) = \sum_{\alpha\subseteq[n],\,|\alpha|=j}\hspace{-0.2cm}x_\alpha
        =\sum_{1\leq\alpha_1<\ldots<\alpha_j\leq n}\hspace{-0.3cm}x_{\alpha_1}\cdots x_{\alpha_i}\in R_{n,j}.
    \end{equation*}
\end{defin}
The same definition applies over any field. It is useful to also define $\es{n}{0}=1$ and $\es{n}{j}=0$ for every $j>n$. %We will discuss the nice properties of the elementary symmetric polynomials $\es{n}{j}$ when needed. 

%Before showing the link between elementary symmetric polynomials and the weight problem, 
We start by recalling Lucas' Theorem~\cite[417-420]{lucas1891theoriedesnombres}, which allows one to compute binomial coefficients modulo a prime number.

\begin{theo}[Lucas' Theorem]\label{thm:lucas}
    Let $p$ be a prime, and let
    \begin{align*}
        m &= m_0 + m_1p + m_2p^2 + \dots + m_rp^r,\\
        n &= n_0 + n_1p + n_2p^2 + \dots + n_rp^r
    \end{align*}
    such that $0\leq m_i, n_i < p$ for every $i=0,\dots,r$ (named \textnormal{p-adic expansion} of $m$ and $n$, or just \textnormal{binary expansion} if $p=2$). Then
    \begin{equation*}
        \binom{m}{n} = \prod_{i=0}^r \binom{m_i}{n_i}\pmod{p}.
    \end{equation*}
\end{theo}

Next, we show how to express the weight of a vector via the values of the elementary symmetric polynomials at its entries.

\begin{lem}\label{lemESPweight}
    Let $n\geq 1$ and $r=\fl*{\operatorname{log}_2 n}$. For $0\leq t\leq n$, let $t = t_0 + t_12 + \dots + t_r2^r$ be the binary expansion of $t$. For a vector $y\in\F_2^n$, the following conditions are equivalent:
    \begin{enumerate}[label=\arabic*)]
        \item $\operatorname{wt}(y)=t$,
        \item $\es{n}{j}(y)=\binom{t}{j}$ for all $j=1,\dots,n$,
        \item $\es{n}{2^i}(y)=t_i$ for all $i=0,\dots,r$.
    \end{enumerate}
\end{lem}
\begin{proof}
    $\textit{"1)}\Rightarrow\textit{2)"}$: Let $y=(y_1,\dots,y_n)\in\F_2^n$ be a vector of weight $t$ and $T\subseteq [n]$ be the set of the indexes corresponding to the non-zero entries of $y$. %(so $|T|=t$). 
    Then
    \begin{align*}
        \es{n}{j}(y)
        &=\sum_{1\leq\alpha_1<\ldots<\alpha_j\leq n}\hspace{-0.5cm}y_{\alpha_1}\cdots y_{\alpha_j}
        =0+\hspace{-0.2cm}\sum_{\substack{1\leq\alpha_1<\ldots<\alpha_j\leq n,\\ \alpha_1,\ldots,\alpha_j\in T}}\hspace{-0.2cm}1=\\
        &=\sum_{\alpha\subseteq T,\,|\alpha|=j}\hspace{-0.2cm}1
        =|\{\alpha\subseteq T\mid |\alpha|=j\}|
        =\binom{t}{j}\pmod{2}.
    \end{align*}
    
    $\textit{"2)}\Rightarrow\textit{3)"}$: By Theorem~\ref{thm:lucas}, for every $i=0,\dots,r$ we have
    \begin{align*}
        \es{n}{2^i}(y)=\binom{t}{2^i}=\binom{t_i}{1}\cdot\prod_{\ell=0,\,\ell\neq i}^r \binom{t_\ell}{0}%=\binom{t_i}{1}\cdot\prod_{\ell=0,\,\ell\neq i}^r\hspace{-0.2cm} 1=\binom{t_i}{1}
        =t_i,
    \end{align*}
    and thus the thesis.
    
    $\textit{"3)}\Rightarrow\textit{1)"}$: Let $y\in\F_2^n$ be such that $\es{n}{2^i}(y)=t_i$ for all $i=0,\dots,r$ and suppose that $\operatorname{wt}(y)=t'$. Thanks to $\textit{"1)}\Rightarrow\textit{2)"}$, we have that $\es{n}{j}(y)=\binom{t'}{j}$ for all $j=1,\dots,n$. Therefore, for every $i\in\{0,\ldots,r\}$, one has
    \begin{align*}
        t_i=\es{n}{2^i}(y)=\binom{t'}{2^i}=t'_i
    \end{align*}
    where the last equality follows from the proof of $\textit{"2)}\Rightarrow\textit{3)"}$. We conclude that $t=t'$.
    
\end{proof}

This result allows us to express the weight of a vector via the evaluations of the elementary symmetric polynomials at its entries, and therefore, to produce a new model for the ESDP over $\F_2$. 

Given a parity check matrix $\mathbf{H}=(h_{i,j})_{i=1,\dots,n-k}^{j=1,\dots,n}$, a syndrome $s=(s_i)_{i=1,\dots,n-k}$ and a target weight $0\leq t\leq n$ with binary expansion $t = \sum\limits_{\ell=0}^r t_\ell2^\ell$, where $r=\fl*{\operatorname{log}_2 n}$, the corresponding instance $(\mathbf{H},s,t)$ of the ESDP can be modelled via the following sets of polynomials in $R_n$:    
\begin{itemize}
    \item \textbf{parity check encoding}:
    \begin{equation*}
        S_1 = \left\{\sum_{j=1}^n h_{i,j}x_j + s_i\mid i=1,\dots,n-k\right\},
    \end{equation*}
    \item \textbf{finite field equations}:
    \begin{equation*}
        S_2 = \{x_j^2+x_j\mid j=1,\dots,n \},
    \end{equation*}
    \item \textbf{weight constraint encoding}:
    \begin{equation*}
        S_3 = \{\es{n}{2^i}(x)+t_i\mid i=0,\dots,r \},\quad r=\lfloor\operatorname{log}_2n\rfloor.
    \end{equation*}
\end{itemize}
Let $S=S_1\cup S_2\cup S_3$. The solutions of the polynomial system $S=0$ coincide with the solutions of the instance $(\mathbf{H},s,t)$ of the ESDP.\\

We now estimate the computational complexity of this model. As discussed in Section~\ref{sec:posso}, it suffices to determine the degree of regularity of $S$.

For that, we consider $S^{\operatorname{top}}=S_1^{\operatorname{top}}\cup S_2^{\operatorname{top}}\cup S_3^{\operatorname{top}}$, where 
\begin{align*}
    S_1^{\operatorname{top}} &= \left\{\sum_{j=1}^n h_{i,j}x_j\mid i=1,\dots,n-k\right\},\\
    S_2^{\operatorname{top}} = \{x_j^2\mid j&=1,\dots,n \},\quad S_3^{\operatorname{top}} = \{\es{n}{2^i}(x)\mid i=0,\dots,r\}.
\end{align*}
Notice that $S^{\operatorname{top}}$ is independent of the target weight $t$. Moreover, $S_2^{\operatorname{top}}$ and $S_3^{\operatorname{top}}$ do not depend on the specific instance $(\mathbf{H},s,t)$ of the ESDP, but only depend on the parameter $n$. For this reason, we first focus on the degree of regularity of $S_2^{\operatorname{top}}\cup S_3^{\operatorname{top}}$.

\begin{notat}
    Let $Q_n=(x_1^2,\dots,x_n^2)\subseteq R_n$ be the ideal generated by the squares of the variables. Let $r=\lfloor\operatorname{log}_2n\rfloor$ and let $I_n=(\es{n}{2^i}\mid i=0,\dots,r)+Q_n\subseteq R_n$ be generated by the polynomials in $S_2^{\operatorname{top}}$ and $S_3^{\operatorname{top}}$. 
\end{notat}

Next we recall some properties of the elementary symmetric polynomials $\es{n}{j}\in R_n$. It is easy to check that they can be computed recursively for $n>1$ using the formulas
\begin{equation}\label{eq:enj_recursion}
    \es{n}{j} =
    \begin{cases}
        \es{n-1}{1} + x_n & \text{if } j = 1, \\
        \es{n-1}{j} + x_n\cdot\es{n-1}{j-1} & \text{if } 2\leq j\leq n-1, \\
        \es{n-1}{n-1}\cdot x_n & \text{if } j = n,
    \end{cases}
\end{equation}
where $\es{n-1}{j}=\hspace{-0.1cm}\sum\limits_{\alpha\subseteq[n-1],\,|\alpha|=j}\hspace{-0.3cm}x_\alpha$ is the elementary symmetric polynomial of degree $j$ in the first $n-1$ variables $x_1,\dots,x_{n-1}$ and $\es{1}{1}$ is equal to $x_1$. %Notice that we could write similar recursive formulas for a different choice of variables.

We also have the following 
\begin{lem}\label{lem:belprodottoesp}
    Let $j,j'\in[n]$, with binary expansions $j = j_0 + j_12 + \dots + j_r2^r$ and $j' = j'_0 + j'_12 + \dots + j'_r2^r$. Then
    \begin{equation*}
        \es{n}{j}\cdot\es{n}{j'}=
        \begin{cases}
            0 & \text{if }j+j'>n\text{ or there exists $\ell$ such that }j_\ell=1=j_\ell',\\
            \es{n}{j+j'} & \text{else},
        \end{cases}
    \end{equation*}
modulo $Q_n$.
\end{lem}
\begin{proof}
    By definition 
    \begin{equation*}
        \es{n}{j}\cdot\es{n}{j'} =\left(\sum_{\alpha\subseteq[n],\,|\alpha|=j}\hspace{-0.3cm}x_\alpha\right)\left(\sum_{\beta\subseteq[n],\,|\beta|=j'}\hspace{-0.3cm}x_\beta\right)
        =\sum_{\substack{\alpha,\beta\subseteq[n],\\|\alpha|=j,\,|\beta|=j'}}\hspace{-0.3cm}x_\alpha x_\beta.
    \end{equation*}
    
    If $j+j'>n$, then $x_\alpha$ and $x_\beta$ share a variable for every $\alpha,\beta\subseteq [n]$. Therefore $x_\alpha x_\beta\in Q_n$ and $\es{n}{j}\cdot\es{n}{j'}=0$ modulo $Q_n$.

    If $j+j'\leq n$, then modulo $Q_n$
    \begin{align*}
        \es{n}{j}\cdot\es{n}{j'}
        &=\sum_{\substack{\alpha,\beta\subseteq[n],\\|\alpha|=j,\,|\beta|=j'}}\hspace{-0.3cm}x_\alpha x_\beta
        =\sum_{\substack{\alpha,\beta\subseteq[n],\,\alpha\,\cap\,\beta=\varnothing,\\|\alpha|=j,\,|\beta|=j'}}\hspace{-0.3cm}x_\alpha x_\beta=\\
        &=\sum_{\gamma\subseteq[n],\,|\gamma|=j+j'}\left(\sum_{\substack{\alpha,\beta\subseteq[n],\,|\alpha|=j,\,|\beta|=j',\\\alpha\,\cap\,\beta=\varnothing,\,\alpha\,\cup\,\beta=\gamma}}\hspace{-0.5cm}x_\alpha x_\beta\right)=\\
        &=\sum_{\gamma\subseteq[n],\,|\gamma|=j+j'}\left(\lvert\{(\alpha,\beta)\mid\substack{\alpha,\beta\subseteq [n],\,|\alpha|=j,\,|\beta|=j',\\\alpha\,\cap\,\beta=\varnothing,\,\alpha\,\cup\,\beta=\gamma}\}\rvert\cdot\, x_\gamma\right)=\\
        &=\sum_{\gamma\subseteq[n],\,|\gamma|=j+j'}\left(\lvert\{\alpha\subseteq [n]\mid|\alpha|=j,\,\alpha\subseteq\gamma\}\rvert\cdot\, x_\gamma\right)=\\
        &=\sum_{\gamma\subseteq[n],\,|\gamma|=j+j'}\binom{j+j'}{j}x_\gamma
        =\binom{j+j'}{j}\es{n}{j+j'}.
    \end{align*}
    By Lucas' Theorem (Theorem~\ref{thm:lucas})
    \begin{equation*}
        \binom{j+j'}{j}=\prod_{i=0}^r\binom{j_i+j'_i}{j_i}.
    \end{equation*}
    Therefore, if there exists $\ell$ such that $j_\ell=1=j_\ell'$, then $\binom{j_\ell+j'_\ell}{j_\ell}=\binom{0}{1}=0$, so $\binom{j+j'}{j}=0$ and $\es{n}{j}\cdot\es{n}{j'}=0$.
    If no such $\ell$ exists, then $\binom{j_i+j'_i}{j_i}=1$ for all $i$, hence $\binom{j+j'}{j}=1$ and $\es{n}{j}\cdot\es{n}{j'}=\es{n}{j+j'}$.
    
\end{proof}

The lemma implies that
\begin{equation}\label{eq:BuildingBlocksESP}
    \es{n}{j}=\es{n}{j_0 + j_12 + \dots + j_r2^r}
    =\es{n}{\sum\limits_{i:j_i=1}\hspace{-0.1cm}2^i}
    =\prod_{i:j_i=1}\es{n}{2^i},
\end{equation}
so $\es{n}{j}\in I_n$ for all $j\in[n]$. It follows that
\begin{equation*}
    I_n=(\es{n}{j}\mid j=1,\dots,n)+Q_n.
\end{equation*}
%which is consistent with the fact that we can express the weight constraint with all the elementary symmetric polynomials as in Lemma~\ref{lemESPweight}.\\

We now want to compute the degree of regularity of $I_n$. Since $I_n$ is generated by homogeneous polynomials, it is a homogeneous ideal, and we can write it as $I_n=I_{n,0}\oplus I_{n,1}\oplus\dots\oplus I_{n,j}\oplus\dots$, where $I_{n,j}=I_n\cap R_{n,j}$ is the $\mathbb{F}_2$-vector space of homogeneous polynomials of $I_n$ of degree $j$. Therefore, the quotient $A_n=R_n/I_n$ is graded via $A_n=A_{n,0}\oplus A_{n,1}\oplus\dots\oplus A_{n,j}\oplus\dots$, where $A_{n,j}=R_{n,j}/I_{n,j}$ is an $\mathbb{F}_2$-vector space. Let $a_{n,j}=\operatorname{dim}_{\F_2} A_{n,j}$.

In order to compute $d_{\operatorname{reg}}(I_n)$, we need to determine the largest $j$ such that $a_{n,j}\neq 0$. We prove a stronger result, some of which may be useful in the sequel. We start with a preliminary lemma.
\begin{lem}\label{lemEspcascata}
    Let $n\geq 3$ and $R_n=\F[x_1,\dots,x_n]$ be the polynomial ring in $n$ variables over a field $\F$. Let $I\subseteq R_n$ be an ideal such that $Q_n=(x_1^2,\dots,x_n^2)\subseteq I$. For every $3\leq m\leq n$ and $2\leq j\leq m-1$,
    \begin{equation*}
        \text{if }\es{m}{j-1}\text{ and }\es{m}{j}\in I\text{, then }\es{m-1}{j}\in I.
    \end{equation*}
\end{lem}
\begin{proof}
    Using \eqref{eq:enj_recursion} we have
    \begin{equation*}
        \es{m}{j}=\es{m-1}{j}+x_m\es{m-1}{j-1}\quad\text{and}\quad\es{m}{j-1}=\es{m-1}{j-1}+x_m\es{m-1}{j-2}.
    \end{equation*}  
    Since $\es{m-1}{0}=1$, we have
    \begin{align*}
        \es{m-1}{j}
        &=\es{m}{j}-x_m\es{m-1}{j-1}=\es{m}{j}-x_m(\es{m}{j-1}-x_m\es{m-1}{j-2})=\\
        &=\underbrace{\es{m}{j}}_{\in I}-x_m\underbrace{\es{m}{j-1}}_{\in I}+\underbrace{x_m^2}_{\in Q_n\subseteq I}\es{m-1}{j-2}\in I.
    \end{align*}    
\end{proof}

The next proposition follows from repeatedly applying Lemma~\ref{lemEspcascata}.
\begin{prop}\label{propAlgoESP}
    Let $n\geq 1$ and let $R_n=\F[x_1,\dots,x_n]$ be the polynomial ring in $n$ variables over a field $\F$. Let $I\subseteq R_n$ be an ideal such that $Q_n=(x_1^2,\dots,x_n^2)\subseteq I$. For every $0\leq u\leq n-1$,
    \begin{equation*}
        \text{if }\es{n}{u+1},\dots,\es{n}{n}\in I\text{, then }I\cap R_{n,j}=R_{n,j}\text{ for }j\geq\fl*{\frac{n+u}{2}}+1.
    \end{equation*}
In particular, $$d_{\operatorname{reg}}(I)\leq\fl*{\frac{n+u}{2}}+1.$$
\end{prop}
\begin{proof}
    The cases $n=1$ and $n=2$ can be checked directly.
    If $n\geq 3$, then we can apply the previous lemma recursively to obtain
    \begin{equation*}
        \begin{array}{c | c c c c c c c c}
             \text{Step 0} & \es{n}{u+1} & \es{n}{u+2} & \es{n}{u+3} & \dots & \es{n}{n-2} & \es{n}{n-1} & \es{n}{n} & \in I\\
             \text{Step 1} & / & \es{n-1}{u+2} & \es{n-1}{u+3} & \dots & \es{n-1}{n-2} & \es{n-1}{n-1} & / & \in I\\
             \text{Step 2} & / & / & \es{n-2}{u+3} & \dots & \es{n-2}{n-2} & / & / & \in I\\
             \vdots &  / & / & / & \vdots & / & / & / & \in I\\
        \end{array}
    \end{equation*}
    For $u+s+1\leq n-s$, that is for $2s\leq n-u-1$, at Step $s$ we have $\es{n-s}{u+s+1},\dots,\es{n-s}{n-s}\in I$. Choosing the largest possible value for $s$, i.e. for $s=\fl*{\frac{n-u-1}{2}}$,
    %\begin{itemize}
        %\item if $n-u \in 2\N$, then $\bar s=\frac{n-u}{2}-1$. In this case, $n-\bar s=n-\frac{n-u}{2}+1=\frac{n+u}{2}+1$ and $u+\bar s+1=u+\frac{n-u}{2}-1+1=\frac{n+u}{2}$, therefore Step $\bar s$ gives us that $\es{\frac{n+u}{2}+1}{\frac{n+u}{2}}$ and $\es{\frac{n+u}{2}+1}{\frac{n+u}{2}+1}\in I$;
        %\item if $n-u\in 2\N+1$, then $\bar s=\frac{n-u-1}{2}$. In this case, $n-\bar s=n-\frac{n-u-1}{2}=\frac{n+u+1}{2}$ and $u+\bar s+1=u+\frac{n-u-1}{2}+1=\frac{n+u+1}{2}=n-\bar s$, therefore Step $\bar s$ gives us that $\es{\frac{n+u+1}{2}}{\frac{n+u+1}{2}}\in I$.
    %\end{itemize}    
    we get that $\es{\fl*{\frac{n+u}{2}}+1}{\fl*{\frac{n+u}{2}}+1}\in (\es{n}{u+1},\dots,\es{n}{n})+Q_n\subseteq I$. Since $(\es{n}{u+1},\dots,\es{n}{n})+Q_n$ is symmetric, any monomial obtained from $\es{\fl*{\frac{n+u}{2}}+1}{\fl*{\frac{n+u}{2}}+1}=x_1\cdots x_{\fl*{\frac{n+u}{2}}+1}$ by permuting the variables still belongs to it, and thus to $I$. It follows that all squarefree monomials of degree $\fl*{\frac{n+u}{2}}+1$ belong to $I$. Since $Q_n\subseteq I$ contains all non-squarefree monomials, every monomial of degree $\geq\fl*{\frac{n+u}{2}}+1$ belongs to $I$, from which the thesis follows.
    
\end{proof}

We are now ready to compute the degree of regularity of $I_n$.

\begin{prop}\label{prop:dregIn}
    For $n\geq1$, let $I_n=(\es{n}{1},\dots,\es{n}{n})+(x_1^2,\dots,x_n^2)\subseteq R_n$ and $A_n=R_n/I_n$. The dimension $a_{n,j}$ of $A_{n,j}$ as an $\F_2$-vector space is
    \begin{equation*}
        a_{n,j}=\begin{cases}
            \binom{n}{j}-\binom{n}{j-1} & \text{if }0\leq j\leq\lfloor\frac{n}{2}\rfloor,\\
            0 & \text{if }j\geq\lfloor\frac{n}{2}\rfloor+1,
        \end{cases}
    \end{equation*}
    where $\binom{n}{-1}=0$. Since $\binom{n}{j}-\binom{n}{j-1}\neq 0$ for all $j\leq\left\lfloor\frac{n}{2}\right\rfloor$, this implies that
    \begin{equation*}
        d_{\operatorname{reg}}(I_n)=\fl*{\frac{n}{2}}+1.
    \end{equation*}
\end{prop}

\begin{proof}
For $j\geq 0$ one has
$$A_{n,j}=\frac{R_{n,j}}{I_{n,j}}\cong\frac{R_{n,j}/Q_{n,j}}{I_{n,j}/Q_{n,j}}$$ hence $$a_{n,j}=\operatorname{dim}_{\F_2}\frac{R_{n,j}}{Q_{n,j}} - \operatorname{dim}_{\F_2}\frac{I_{n,j}}{Q_{n,j}},$$
where $Q_{n,j}=Q_n\cap R_{n,j}$. Since $R_{n,j}/Q_{n,j}$ is generated by all squarefree monomials of degree $j$, 
    $$\operatorname{dim}_{\F_2}\frac{R_{n,j}}{Q_{n,j}}=\binom{n}{j}.$$
The quotient $I_{n,j}/Q_{n,j}$ is generated by all products of the form $x_\alpha\cdot\es{n}{m}$, where $1\leq m\leq j$ and $x_\alpha$ is a squarefree monomial of degree $j-m$. However, in general, these generators are not linearly independent. In order to compute the dimension of $I_{n,j}/Q_{n,j}$, we consider an \textit{incidence matrix} $M_{n,j}$ where:
\begin{itemize}
    \item the rows are indexed by the generators $x_\alpha\cdot\es{n}{m}$ described before,
    \item the columns are indexed by the squarefree monomials  of degree $j$ in $n$ variables,
    \item the entry in row $x_\alpha\cdot\es{n}{m}$ and column $x_\beta$ is the coefficient of the monomial $x_\beta$ in the polynomial $x_\alpha\cdot\es{n}{m}$.
\end{itemize}
By construction, the rank of $M_{n,j}$ is the dimension of $I_{n,j}/Q_{n,j}$ as an $\F_2$-vector space. Since $\operatorname{deg}x_\alpha +m=j$ and modulo $Q_n$, we have 
\begin{equation*}
    x_\alpha\cdot\es{n}{m}(x)
    =\sum_{\gamma\subseteq[n],\,|\gamma|=m}\hspace{-0.3cm}x_\alpha x_\gamma
    =\sum_{\substack{\gamma\subseteq[n],\,|\gamma|=m, \\ \gamma\,\cap\,\alpha=\varnothing}} \hspace{-0.3cm}x_\alpha x_\gamma
    =\sum_{\substack{\beta\subseteq[n],\,|\beta|=j, \\ \alpha\,\subseteq\,\beta}}\hspace{-0.3cm}x_\beta.
\end{equation*}
In other words, the coefficient of $x_\beta$ in $x_\alpha\cdot\es{n}{m}$ is $1$ if and only if $\alpha\subseteq\beta$. Therefore, the matrix $M_{n,j}$ is the incidence matrix relative to the inclusion of the subsets $\alpha\subseteq[n]$ of cardinality $0\leq|\alpha|\leq j-1$ in the subsets $\beta\subseteq[n]$ of cardinality $j$. For example, if $n=4$ and $j=3$, we get the incidence matrix
\begin{equation*} 
    \begin{array}{c || c | c c c || c}
        M_{4,3}\text{ via }\cdot & x_1x_2x_3 & x_1x_2x_4 & x_1x_3x_4 & x_2x_3x_4 & \\
        \hline
        \hline
        1\cdot\es{4}{3} & 1 & 1 & 1 & 1 & \varnothing \\
        x_1\cdot\es{4}{2} & 1 & 1 & 1 & 0 & \{1\} \\
        x_2\cdot\es{4}{2} & 1 & 1 & 0 & 1 & \{2\} \\
        x_3\cdot\es{4}{2} & 1 & 0 & 1 & 1 & \{3\} \\
        \hline
        x_1x_2\cdot\es{4}{1} & 1 & 1 & 0 & 0 & \{1,2\} \\
        x_1x_3\cdot\es{4}{1} & 1 & 0 & 1 & 0 & \{1,3\} \\
        x_2x_3\cdot\es{4}{1} & 1 & 0 & 0 & 1 & \{2,3\} \\
        \hline
        x_4\cdot\es{4}{2} & 0 & 1 & 1 & 1 & \{4\} \\
        x_1x_4\cdot\es{4}{1} & 0 & 1 & 1 & 0 & \{1,4\} \\
        x_2x_4\cdot\es{4}{1} & 0 & 1 & 0 & 1 & \{2,4\} \\
        x_3x_4\cdot\es{4}{1} & 0 & 0 & 1 & 1 & \{3,4\} \\
        \hline
        \hline
         & \{1,2,3\} & \{1,2,4\} & \{1,3,4\} & \{2,3,4\} &M_{4,3}\text{ via }\subseteq \\
    \end{array}
\end{equation*}
\begin{comment}
which is just
\begin{equation*} 
    \begin{array}{c | c | c c c}
        M_{4,3}\text{ with }\subseteq & \{1,2,3\} & \{1,2,4\} & \{1,3,4\} & \{2,3,4\} \\
        \hline
        \varnothing & 1 & 1 & 1 & 1 \\
        \{1\} & 1 & 1 & 1 & 0 \\
        \{2\} & 1 & 1 & 0 & 1 \\
        \{3\} & 1 & 0 & 1 & 1 \\
        \hline
        \{1,2\} & 1 & 1 & 0 & 0 \\
        \{1,3\} & 1 & 0 & 1 & 0 \\
        \{2,3\} & 1 & 0 & 0 & 1 \\
        \hline
        \{4\} & 0 & 1 & 1 & 1 \\
        \{1,4\} & 0 & 1 & 1 & 0 \\
        \{2,4\} & 0 & 1 & 1 & 0 \\
        \{3,4\} & 0 & 0 & 1 & 1 \\
    \end{array}
\end{equation*}
\end{comment}

%If we first isolate the monomials [resp. sets] containing $x_n$ [resp. $n$] and then order the rows and columns in the proper way, we can notice, as in the example, that the matrix $M_{n,j}$ has a nice structure, with some repeating submatrices, and also a null and an identity submatrix. This structure was exploited by 
Wilson~\cite[Lemma 1]{WILSON1990609} proved that \begin{equation*}
    \operatorname{rk}(M_{n,j})=\binom{n}{j-1}\qquad\text{for every }j\leq\fl*{\frac{n}{2}}.
\end{equation*} 
This implies that $a_{n,j}=\binom{n}{j}-\binom{n}{j-1}$ for every $j\leq\fl*{\frac{n}{2}}$, as desired.\\

To conclude the proof of Proposition~\ref{prop:dregIn}, it remains to prove that $a_{n,j}=0$ for every $j\geq\fl*{\frac{n}{2}}+1$, that is, to prove that $I_{n,j}=R_{n,j}$ for every $j\geq\fl*{\frac{n}{2}}+1$.

Finally, the ideal $I_n$ satisfies the hypothesis of Proposition~\ref{propAlgoESP} with $\F=\F_2$ and $u=0$, and therefore $I_{n,j}=R_{n,j}$ for every $j\geq\fl*{\frac{n}{2}}+1$, which concludes the proof.
\end{proof}

The previous results allow us to give a first estimate of $d_{\operatorname{sol}}(S)$, where, we recall, the solutions of $S=0$ correspond to the solutions of the related instance of the ESDP.
\begin{theo}\label{cordsolS}
    For $n\geq 2$, the solving degree of $S$ is bounded from above by
    \begin{equation}
        d_{\operatorname{sol}}(S)\leq\begin{cases}
            2^r+1 & \text{if }n=2^{r+1}-2\text{ or }n=2^{r+1}-1,\\
            2^r & \text{if }r\geq 2\text{ and }2^r\leq n\leq 2^{r+1}-3,\\
        \end{cases}
    \end{equation}
    where $r=\fl*{\operatorname{log}_2 n}\geq 1$.
\end{theo}
\begin{proof}
    By Theorem~\ref{thm:salizz}, we have that $d_{\operatorname{sol}}(S)\leq\max\{d_{\operatorname{reg}}(S)+1,\underset{f\in S}{\max}\operatorname{deg}(f)\}$, where
    \begin{equation*}
        d_{\operatorname{reg}}(S)+1\leq d_{\operatorname{reg}}(S_2\cup S_3)+1=d_{\operatorname{reg}}(I_n)+1=\fl*{\frac{n}{2}}+2
    \end{equation*}
    by Proposition~\ref{prop:dregIn} and
    \begin{equation*}
        \underset{f\in S}{\max}\operatorname{deg}(f)=\underset{f\in S_1\cup S_2\cup S_3}{\max}\operatorname{deg}(f)=\operatorname{max}\{1,2,2^0,\dots,2^r\}=2^r
    \end{equation*}
    since $r\geq 1$. We distinguish two cases:
    \begin{itemize}
        \item If $r\geq 2$ and $2^r\leq n\leq 2^{r+1}-3$, then 
        $$\fl*{\frac{n}{2}}+2\leq\fl*{2^r-\frac{3}{2}}+2=2^r-2+2=2^r.$$
        \item If $n\in\{2^{r+1}-2, 2^{r+1}-1\}$, then
        $$\fl*{\frac{n}{2}}+2=2^r-1+2=2^r+1> 2^r.$$
        Notice that this includes the case $r=1$, since in that case $n\in\{2,3\}$.
    \end{itemize}
    This completes the proof of the statement.
    
\end{proof}

Comparing this result with Theorem~\ref{teoSQM} we see that, for non-trivial values of $t$, the complexity of solving our polynomial model $S$ for the ESDP is smaller than the complexity of solving the model from~\cite{caminata2025quadraticmodelingssyndromedecoding}. \\

%\subsection{A refinement of the model}\label{sezFurther}
%We now recall that, until now, the parity check encoding has not been involved yet. This is because no matter which instance of the ESDP we consider, in most cases the maximum value in the bound given by Theorem~\ref{thm:salizz} is achieved by the degree of the highest degree polynomial of $S$, that is, $\es{n}{2^r}$.\\

Notice that the bound of Theorem~\ref{cordsolS} is obtained by considering only the system consisting of the field equations and the equations corresponding to the weight constraint. Since the dependence on the weight $t$ is eliminated when passing to the top-dimensional part $S^{\operatorname{top}}$ of the system, the bound depends only on $n$. In particular, it is independent of the underlying code and of the instance $(\mathbf{H},s,t)$ of the ESDP. This suggests that one may further improve our model using the parity check encoding. In the next section, we propose a variant of the model that depends on the instance of the ESDP.

\subsection{An improved model and its computational complexity}\label{secComplexityESDP}

To improve the model described in the previous section, notice that the bound given in Theorem~\ref{cordsolS} is often controlled by the degrees of the equations. 
Therefore, we propose to remove the highest-degree polynomial from the system $S=0$. 

We now analyze the resulting system. 
The next lemma is a variation of Lemma~\ref{lemESPweight} and is provided for completeness.
\begin{lem}\label{lemNSESDP}
    Given $n\geq1$, $r=\fl*{\operatorname{log}_2 n}$ and $0\leq t\leq n$ with binary expansion $t = \sum\limits_{\ell=0}^r t_\ell2^\ell$, let
    \[
        S_3'=\{\es{n}{2^i}(x)+t_i\mid i=0,\dots,r-1\}.
    \]
    The vectors of $\F_2^n$ that satisfy the polynomial system $S_3'=0$ are exactly
    \begin{equation*}
        \begin{cases}
            \text{the vectors of weight }t\text{ or }t+2^r&\text{if }0\leq t\leq n-2^r\\
            \text{the vectors of weight }t&\text{if }n-2^r+1\leq t\leq 2^r-1\\
            \text{the vectors of weight }t\text{ or }t-2^r&\text{if }2^r\leq t\leq n.
        \end{cases}
    \end{equation*}    
\end{lem}
\begin{proof}
    This easily follows from Lemma~\ref{lemESPweight}:
    \begin{itemize}
        \item If $0\leq t\leq n-2^r$, then $t_r=0$. Let $t'=t+2^r$, then $2^r\leq t'\leq n$. Therefore, a solution $y\in\F_2^n$ of $S_3'=0$, can have either $\operatorname{wt}(y)=t$, if $\es{n}{2^r}(y)=0$, or $\operatorname{wt}(y)=t'=t+2^r$, if $\es{n}{2^r}(y)=1$.
        \item If $n-2^r+1\leq t\leq 2^r-1$, we have $t_r=0$. Let $y$ be a solution of $S_3'=0$. Then $\es{n}{2^r}(y)$ can be either $0$ or $1$. If $\es{n}{2^r}(y)=0$, then by Lemma~\ref{lemESPweight}, its weight is exactly $t$. If $\es{n}{2^r}(y)=1$, $y$ has weight $t+2^r\geq n+1$, which is a contradiction. Therefore, all solutions of $S_3'=0$ must have weight $t$.
        \item If $2^r\leq t\leq n$, we can make similar considerations as in the first case with $t''=t-2^r$, which is such that $0\leq t''\leq n-2^r$.
    \end{itemize}
\end{proof}

%We could also ignore more high degree polynomials, such as $\es{n}{2^{r-1}}$, but we will have for each $t$ more possible weights among the solutions, and not even a set of parameters $t$ where the weight of the solutions is unique.\\

We can now consider the improved model $S'=S_1\cup S_2\cup S'_3$. Notice that the solutions of the instance $(\mathbf{H},s,t)$ of the ESDP are also solutions of the polynomial system $S'=0$. We recall that the ESDP has a unique solution if $t\leq\fl*{\frac{d(\mathbf{C})-1}{2}}\leq\fl*{\frac{n-1}{2}}\leq 2^r-1$. Notice that this falls in the two first cases of Lemma~\ref{lemNSESDP}.
Under the cryptographic assumption that the solution of an instance $(\mathbf{H},s,t)$ of the ESDP is unique, there are two possible scenarios:
\begin{itemize}
    \item The solution of $S'=0$ is unique, hence it is the solution of the ESDP. This occurs if the \textit{wrong} instance $(\mathbf{H},s,t\pm 2^r)$ happens to have no solution, or if the target weight $t$ is such that $n-2^r+1\leq t\leq 2^r-1$. For example, if $n$ is a power of $2$, such as in the Parameter set $9$ and $10$ of~\cite{ClassicMcEliece2022}, all the non-trivial choices for $t$ are in this interval.
    \item The system $S'=0$ has more than one solution, but it contains exactly one vector of weight $t$. In this situation, an attacker can efficiently recover the desired vector by computing the Hamming weight of the solutions of $S'=0$ and discarding the ones with the wrong weight. This final selection step is computationally negligible and does not contribute to the overall complexity of the model.
\end{itemize}

%\subsection{Computational complexity of the model}\label{secComplexityNSESDP}
In order to compute the complexity of solving the ESDP using this new polynomial model, we want to estimate the solving degree of the polynomial system $S'=0$. As we have done for the previous model, we start by bounding the degree of regularity $d_{\operatorname{reg}}(S')$. 
We start by studying the ideal 
$$J_n=(\es{n}{2^0},\dots,\es{n}{2^{r-1}})+Q_n,$$ 
which is generated by the polynomials in $S_2^{\operatorname{top}}$ and $(S_3')^{\operatorname{top}}$. As for $I_n$, by Lemma~\ref{lem:belprodottoesp}
\begin{equation*}
    J_n=(\es{n}{j}\mid j=1,\dots,2^r-1)+Q_n.
\end{equation*}
The next lemma will be useful in the sequel.
\begin{lem}\label{full prod}
    Let $k\leq n$.
    For every $j=1,\dots,n$, one has
    \begin{align*}
        \es{n}{j}(x_1,\dots,x_n)=\sum_{i=0}^j\es{k}{i}(x_1,\dots,x_k)\cdot\es{n-k}{j-i}(x_{k+1},\dots,x_n),
    \end{align*}
    where $\es{m}{0}=1$ and $\es{m}{l}=0$ for $l>m$.
    %, for both $m=k$ and $m=n-k$.
\end{lem}
\begin{proof}
    If $i>k$ or $j-i>n-k$, the product $\es{k}{i}(x_1,\dots,x_k)\cdot\es{n-k}{j-i}(x_{k+1},\dots,x_n)$ is equal to $0$.
    If $i,j$ are such that $i\leq k$ and $j-i\leq n-k$, then
    \begin{align*}
        \es{k}{i}(x_1,\dots,x_k)&\cdot\es{n-k}{j-i}(x_{k+1},\dots,x_n)=\\
        &=\left(\sum_{\alpha\subseteq[k],\,|\alpha|=i}\hspace{-0.2cm}x_\alpha\right)\left(\sum_{\beta\subseteq\{k+1,\dots,n\},\,|\beta|=j-i}\hspace{-0.3cm}x_\beta\right)=\\
        &=\sum_{\substack{\alpha\subseteq[k],\,|\alpha|=i,\\\beta\subseteq\{k+1,\dots,n\},\,|\beta|=j-i}}\hspace{-0.4cm}x_\alpha x_\beta
        =\sum_{\substack{\gamma\subseteq[n],\,|\gamma\,\cap\,[k]|=i,\\|\gamma\,\cap\,\{k+1,\dots,n\}|=j-i}}\hspace{-0.3cm}x_\gamma,
    \end{align*}
hence
    \begin{align*}
        &\sum_{i=0}^j\es{k}{i}(x_1,\dots,x_k)\cdot\es{n-k}{j-i}(x_{k+1},\dots,x_n)=\\
        =\sum_{i=0}^j&\sum_{\substack{\gamma\subseteq[n],\,|\gamma\,\cap\,[k]|=i,\\|\gamma\,\cap\,\{k+1,\dots,n\}|=j-i}}\hspace{-0.4cm}x_\gamma
        =\sum_{\gamma\subseteq[n],\,|\gamma|=j}\hspace{-0.2cm}x_\gamma
        =\es{n}{j}(x_1,\dots,x_n).
    \end{align*}
\end{proof}

The lemma allows us to express the ideal $J_n$ in different ways.
\begin{lem}\label{JnKnK'n}
For any $1\leq k\leq n-1$, the following ideals are equal:
    \begin{align*}
        J_n&=(\es{n}{2^i}(x)\mid i=0,\dots,r-1)+Q_n \\ 
        &=(\es{n}{j}(x)\mid j=1,\dots,2^r-1)+Q_n\\
        &=(\es{k}{2^i}(x_1,\dots,x_k)+\es{n-k}{2^i}(x_{k+1},\dots,x_n)\mid i=0,\dots,r-1)+Q_n\\
        &=(\es{k}{j}(x_1,\dots,x_k)+\es{n-k}{j}(x_{k+1},\dots,x_n)\mid j=1,\dots,2^r-1)+Q_n.
    \end{align*}
\end{lem}
\begin{proof}
    Throughout the proof, we denote $$K_n=(\es{k}{2^i}(x_1,\dots,x_k)+\es{n-k}{2^i}(x_{k+1},\dots,x_n)\mid i=0,\dots,r-1)+Q_n,$$ $$K'_n=(\es{k}{j}(x_1,\dots,x_k)+\es{n-k}{j}(x_{k+1},\dots,x_n)\mid j=1,\dots,2^r-1)+Q_n.$$ We also write $\es{k}{j}$ for $\es{k}{j}(x_1,\dots,x_k)$, and similarly for $\es{n}{j}$ and $\es{n-k}{j}$.

    The second equality follows from Lemma~\ref{lem:belprodottoesp}.
    It is clear that $K_n\subseteq K_n'$. Hence, to prove the chain of equalities in the statement, it suffices to prove the inclusions $K_n'\subseteq J_n$ and $J_n\subseteq K_n$.
    
    $K_n'\subseteq J_n$: We prove that $\es{k}{j}+\es{n-k}{j}\in J_n$ for all $j=1,\dots,2^r-1$ by induction on $j$. Since $\es{k}{j}+\es{n-k}{j}=0$ for $j>\max\{k,n-k\}$, we may restrict to $j\leq\max\{k,n-k\}$.
    The case $j=1$ is trivial, since $\es{k}{1}+\es{n-k}{1}=\es{n}{1}\in J_n$. For $j>1$, we consider two cases.
    \begin{itemize}
        \item If $j\leq\min\{k,n-k\}$, by Lemma~\ref{full prod} we have
        \begin{equation*}
            \es{n}{j}=\sum\limits_{i=0}^j\es{k}{i}\cdot\es{n-k}{j-i}=\es{k}{j}+\sum\limits_{i=1}^{j-1}\es{k}{i}\cdot\es{n-k}{j-i}+\es{n-k}{j}.
        \end{equation*}
        Since $\es{n}{j}\in J_n$, it suffices to show that $q_j=\sum\limits_{i=1}^{j-1}\es{k}{i}\cdot\es{n-k}{j-i}\in J_n$.

        Let $p_i=\es{k}{i}+\es{n-k}{i}$. One has
        \begin{equation*}
            q_j=\sum\limits_{i=1}^{j-1}\es{k}{i}(p_{j-i}+\es{k}{j-i})=\sum\limits_{i=1}^{j-1}\es{k}{i}\cdot p_{j-i}+\sum\limits_{i=1}^{j-1}\es{k}{i}\cdot\es{k}{j-i}.
        \end{equation*}
        The first sum, that we denote by $p'_j$, belongs to $J_n$, since $p_i\in J_n$ for all $i=1,\dots,j-1$ by the induction hypothesis. By Lemma~\ref{lem:belprodottoesp}, the second sum equals $\alpha\es{k}{j}$, for some $\alpha\in\mathbb{F}_2$.
        If $\alpha=0$, then $q_j=p_j'\in J_n$. If $\alpha=1$, then $q_j=p'_j+\es{k}{j}$, and therefore $\es{n}{j}=\es{k}{j}+q_j+\es{n-k}{j}=p'_j+\es{n-k}{j}$, which implies that $\es{n-k}{j}\in J_n$. In addition,
        \begin{equation*}
            q_j=\sum\limits_{i=1}^{j-1}(p_i+\es{n-k}{i})\es{n-k}{j-i}=\sum\limits_{i=1}^{j-1}p_i\cdot\es{n-k}{j-i}+\sum\limits_{i=1}^{j-1}\es{n-k}{i}\cdot\es{n-k}{j-i}.
        \end{equation*}
        Let $p''_j$ denote the first sum, then $p''_j\in J_n$ as before and the second sum is $\alpha\es{n-k}{j}$, for the same $t\in\F_2$ as above, since the bit check of Lemma~\ref{lem:belprodottoesp} involves the same pairs $(i,j-i)$. Since $\alpha=1$, we obtain $q_j=p''_j+\es{n-k}{j}\in J_n$. This concludes the proof. 

        \item If $j>\min\{k,n-k\}$, up to exchanging $k$ and $n-k$ we may assume without loss of generality that $k<j\leq n-k$. We want to prove that $\es{k}{j}+\es{n-k}{j}=\es{n-k}{j}\in J_n$. We proceed as before:
        \begin{equation*}
            \es{n}{j}=\sum\limits_{i=0}^j\es{k}{i}\cdot\es{n-k}{j-i}\overset{j>k}{=}\sum\limits_{i=1}^k\es{k}{i}\cdot\es{n-k}{j-i}+\es{n-k}{j}.
        \end{equation*}
        Again, letting $p_i=\es{k}{i}+\es{n-k}{i}$, we have
        \begin{equation*}
            \sum\limits_{i=1}^k\es{k}{i}\cdot\es{n-k}{j-i}=\sum\limits_{i=1}^k\es{k}{i}(p_{j-i}+\es{k}{j-i})=\sum\limits_{i=1}^k\es{k}{i}\cdot p_{j-i}+\sum\limits_{i=1}^k\es{k}{i}\cdot\es{k}{j-i},
        \end{equation*}
        where the first summand belongs to $J_n$, since by induction $p_i\in J_n$ if $i<j$, and the second summand is $0$ by Lemma~\ref{lem:belprodottoesp}. Therefore, $\es{n-k}{j}\in J_n$ as desired.
    \end{itemize}

    $J_n\subseteq K_n$: We want to prove that $\es{n}{2^i}\in K_n$ for $i=0,\dots,r-1$. We proceed by induction on $i$. The case $i=0$ is clear. Hence, suppose $i>0$ and assume that $\es{n}{2^j}\in K_n$ for $0\leq j\leq i-1$. By Lemma~\ref{lem:belprodottoesp}, this implies $\es{n}{j}\in K_n$ for $1\leq j\leq 2^i-1$. By Lemma~\ref{full prod}, moreover,
    \begin{equation*}
        \es{n}{2^i}=\sum\limits_{j=0}^{2^i}\es{k}{j}\cdot\es{n-k}{2^i-j}=\es{k}{2^i}+\sum\limits_{j=1}^{2^i-1}\es{k}{j}\cdot\es{n-k}{2^i-j}+\es{n-k}{2^i}.
    \end{equation*}
    So it suffices to prove that $\sum\limits_{j=1}^{2^i-1}\es{k}{j}\cdot\es{n-k}{2^i-j}\in K_n$.

    Let $1\leq u\leq\min\{k,2^i-1\}$, then by Lemma~\ref{full prod}
    \begin{align*}
        K_n\ni\es{k}{u}\cdot\es{n}{2^i-u}&=\es{k}{u}\sum\limits_{j=0}^{2^i-u}\es{k}{j}\cdot\es{n-k}{2^i-u-j}=\\
        %&=\es{k}{u}\cdot\es{k}{2^i-u}+\es{k}{u}\sum\limits_{j=0}^{2^i-u-1}\es{k}{j}\cdot\es{n-k}{2^i-u-j}\overset{j'=j+u}{=}\\
        &=\sum\limits_{j'=u}^{2^i-1}\es{k}{u}\cdot\es{k}{j'-u}\cdot\es{n-k}{2^i-j'}.
    \end{align*}
    The last equality follows from setting $j'=j+u$ and by observing that $\es{k}{u}\cdot\es{k}{2^i-u}=0$ for every $u$ by Lemma~\ref{lem:belprodottoesp}, since $u$ and $2^i-u$ must share at least one non-zero bit in their binary expansion. 
    %as their sum $2^i$ has all bits equal to zero, except for the last one 
    Therefore
    \begin{align*}
        K_n\ni\sum\limits_{u=1}^{\min\{k,2^i-1\}}\es{k}{u}\cdot\es{n}{2^i-u}
        &=\sum\limits_{u=1}^{\min\{k,2^i-1\}}\sum\limits_{j'=u}^{2^i-1}\es{k}{u}\cdot\es{k}{j'-u}\cdot\es{n-k}{2^i-j'}=\\
        &=\sum\limits_{j'=1}^{2^i-1}\sum\limits_{u=1}^{j'}\es{k}{u}\cdot\es{k}{j'-u}\cdot\es{n-k}{2^i-j'}=\\
        &=\sum\limits_{j'=1}^{2^i-1}|T_{j'}|\cdot\es{k}{j'}\cdot\es{n-k}{2^i-j'},
    \end{align*}
    where $T_{j'}$ is defined, following Lemma~\ref{lem:belprodottoesp}, as the set
    \begin{align*}
        T_{j'} &=\{u\in[1,j']\mid \text{there is no $\ell$ such that $u_\ell=1=(j'-u)_\ell$}\}=\\
        &=\{u\in[1,j']\mid \text{ if } j'_\ell =0 \text{ for some }\ell, \text{ then } t_\ell=0\}.
    \end{align*}
    Notice that if $j'>k$, then $\es{k}{j'}=0$ and $T_{j'}$ plays no role.   
    In addition, for every $j'$
    \begin{equation*}
        |T_{j'}|=2^{|\{\ell\in[0,\lceil\operatorname{log}_2(j')\rceil]\mid j'_\ell =1\}|}-1,
    \end{equation*}
    since, if $j'_\ell =1$, then $u_\ell$ can be either $0$ or $1$, however $0\not\in T_{j'}$. In particular, $|T_{j'}|$ is odd for all $j'\leq k$. It follows that $$\sum\limits_{j=1}^{2^i-1}\es{k}{j}\cdot\es{n-k}{2^i-j}=\sum\limits_{u=1}^{\min\{k,2^i-1\}}\es{k}{u}\cdot\es{n}{2^i-u}\in K_n,$$  
    which concludes the proof.
\begin{comment}
    if $u\in T_{j'}$, then the non-zero bits in the binary expansion of $j'-t$ and $t$ must correspond to each other as summarized in the following table:
    \begin{equation*}
        \begin{array}{ c | c c c c }
            &\multicolumn{4}{c}{\text{Bit occurrence}}\\
            \hline\hline
            j'-t & 0 & 0 & 1 & \cancel{1}\\
            t & 0 & 1 & 0 & \cancel{1}\\
            \hline
            j' & 0 & 1 & 1 & \cancel{0}
        \end{array}
    \end{equation*}
    Therefore, in the binary sum $(j'-t)+t=j'$, there is no remainder and the set $T_{j'}$ is just
\end{comment}

\end{proof}

Before continuing with the estimate of $d_{\operatorname{reg}}(S')$, notice that one could use the parity check encoding to eliminate some variables. In fact, since $\mathbf{H}\in\F_2^{(n-k)\times n}$ is a full rank matrix, it can be written in the form 
\begin{equation*}
    \mathbf{H}=(\mathbf{L}\mid \mathbf{I}_{n-k})=
    \left(\begin{array}{c : c}
        \cdots\ell_{1}\cdots & \\
        \vdots & \quad \mathbf{I}_{n-k}\quad \\
        \cdots\ell_{n-k}\cdots & \\
    \end{array}\right),
\end{equation*}%$(L\mid \mathbf{I}_{n-k})$
where $\mathbf{I}_{n-k}$ is the identity matrix of size $n-k$ and $\mathbf{L}$ is a matrix in $\F_2^{(n-k)\times k}$, by performing a Gaussian elimination on the rows and, if necessary, swapping columns. Notice that these operations do not change the underlying problem, since the matrix obtained is the parity check matrix of an equivalent code. We can therefore assume that $\mathbf{H}$ is of the form above. With this assumption, the equations corresponding to $S_1$ are of the form
\begin{align*}
        x_{k+1}=\ell_1(x_1,\dots,x_k)+s_1,\;\dots,\;x_n=\ell_{n-k}(x_1,\dots,x_k)+s_{n-k}.
\end{align*}
\begin{comment}

By substituting $x_{k+1},\dots,x_n$ into the other equations of $S'=0$, we get an equivalent polynomial system $T=0$, whose degree of regularity is bounded by $k+1$. 
In fact, $(T^{\operatorname{top}})\subseteq R_k$ contains $Q_k=(x_1^2,\ldots,x_k^2)$, hence $d_{\operatorname{reg}}(T)\leq d_{\operatorname{reg}}(Q_k)=k+1$.

Notice that even though $d_{\operatorname{reg}}(T)<d_{\operatorname{reg}}(S)$ whenever $k\leq\fl*{\frac{n}{2}}$, if we proceed as in Theorem~\ref{cordsolS} we may not obtain a smaller bound for $d_{\operatorname{sol}}(T)$. This is due to the fact that the bound is controlled by the degrees of the equations, which is the same for $S$ and $T$. 

\end{comment}
Hence, $S_1^{\operatorname{top}}$ consists of the linear polynomials $x_j+\ell_{j-k}(x_1,\dots,x_k)$ for $j=k+1,\dots,n$. Notice that
\begin{equation*}
    \frac{R_n}{J_n+(S_1^{\operatorname{top}})}=\frac{R_n}{J_n+(x_{k+1}+\ell_1(x_1,\dots,x_k),\dots,x_n +\ell_{n-k}(x_1,\dots,x_k))}\cong\frac{R_k}{J'_k},
\end{equation*}
where
\begin{equation*}
    J'_k=(\es{k}{2^i}(x_1,\dots,x_k)+\es{n-k}{2^i}(\ell_1,\dots,\ell_{n-k})\mid i=0,\dots,r-1)+Q_k
\end{equation*}
and $\ell_j=\ell_j(x_1,\dots,x_k)$. Therefore, since the ideal generated by $(S')^{\operatorname{top}}$ is precisely $J_n+(S_1^{\operatorname{top}})$, we have that $d_{\operatorname{reg}}(S')=d_{\operatorname{reg}}(J'_k)$, and for this reason we can now focus on the ideal $J'_k$. 

\begin{notat}
    For $j=1,\dots,n-k$, we write $\ell_j$ for $\ell_j(x_1,\dots,x_k)$. For every $\alpha=\{\alpha_1,\dots,\alpha_i\}\subseteq [n-k]$, $\ell_\alpha$ denotes the product $\ell_{\alpha_1}\cdots\ell_{\alpha_i}\in R_{k,i}$. Finally, for $g\in R_{n-k}$, we abbreviate $g(\ell_1,\dots,\ell_{n-k})$ as $g(\ell)$.
\end{notat}

The substitution reduces the degrees of the involved polynomials, since $\es{k}{2^i}(x)$ and $\es{n-k}{2^i}(\ell)$ are $0$ if $2^i$ is greater than $k$ and $n-k$, respectively.
Moreover, it highlights and isolates the dependence on the matrix $\mathbf{L}$, and so on the specific instance of the ESDP. The next lemma clarifies that dependence.

%The following lemma gives us a way to make use of this, recalling that the $k$ columns of the matrix $\mathbf{L}$ are associated with the variables $x_1,\dots,x_k$:
\begin{lem}
    Let $k\leq n$ and $1\leq i\leq n-k$. For any $\beta\subseteq[k]$ of cardinality $i$, the coefficient of the squarefree monomial $x_\beta$ in the polynomial \begin{equation*}
        \es{n-k}{i}(\ell)=\sum\limits_{\alpha\subseteq [n-k],\,|\alpha|=i}\hspace{-0.2cm}\ell_\alpha\in R_{k,i}
    \end{equation*}
    is the sum of all $i$-minors of $\mathbf{L}$ whose columns are indexed by $\beta$.
\end{lem}
\begin{proof}
    Let $\mathbf{L}=(\ell_{i,j})_{i=1,\dots,n-k}^{j=1,\dots,k}$. For any $\alpha=\{\alpha_1,\dots,\alpha_i\}\subseteq[n-k]$ and $\beta=\{\beta_1,\ldots,\beta_i\}\subseteq [k]$, the determinant of the submatrix of $\mathbf{L}$ consisting of rows indexed by $\alpha$ and columns indexed by $\beta$ is 
    \begin{equation*}
        \operatorname{det}(\mathbf{L}[\alpha,\beta])=\sum_{\sigma\in S_i}\prod_{j=1}^i \ell_{\alpha_j,\beta_{\sigma(j)}}.
    \end{equation*}
    We want to prove that the coefficient of $x_\beta$ in $\es{n-k}{i}(\ell)$ is the sum of all $\operatorname{det}(\mathbf{L}[\alpha,\beta])$, with $\alpha\subseteq [n-k]$ and $|\alpha|=i$.

    Consider a summand $\ell_\alpha$ of $\es{n-k}{i}(\ell)$. We have
    \begin{align*}
        \ell_\alpha
        &=\ell_{\alpha_1}(x_1,\dots,x_k)\cdots\ell_{\alpha_i}(x_1,\dots,x_k)=\\
        &=(\ell_{\alpha_1,1}x_1+\cdots+\ell_{\alpha_1,k}x_k)\cdots(\ell_{\alpha_i,1}x_1+\cdots+\ell_{\alpha_i,k}x_k).
    \end{align*}
    The monomial $x_\beta$ is equal to the product $x_{\beta_{\sigma(1)}}\cdots x_{\beta_{\sigma(i)}}$ for every permutation $\sigma\in S_i$. For a given permutation $\sigma$, the coefficient of the product $x_{\beta_{\sigma(1)}}\cdots x_{\beta_{\sigma(i)}}$ is $\prod\limits_{j=1}^i\ell_{\alpha_j,\beta_{\sigma(j)}}$. Therefore, the coefficient of $x_\beta$ in $\ell_\alpha$ is $\sum\limits_{\sigma\in S_i}\prod\limits_{j=1}^i\ell_{\alpha_j,\beta_{\sigma(j)}}=\operatorname{det}(\mathbf{L}[\alpha,\beta])$. The thesis follows from observing that $\es{n-k}{i}(\ell)=\sum\limits_{\alpha\subseteq [n-k],\,|\alpha|=i}\hspace{-0.2cm}\ell_\alpha$.\\
\end{proof}

This lemma gives us a way to compute $\es{n-k}{2^i}(\ell)$, but more importantly, it gives us a way to quickly see if some of the $\es{n-k}{2^i}(\ell)$ can be discarded. In fact, if $2^i>\operatorname{rk}(\mathbf{L})$, then every minor of order $2^i$ of the matrix $\mathbf{L}$ vanishes, hence $\es{n-k}{2^i}(\ell)\in Q_k$. Therefore
\begin{align*}
    J'_k&=(\es{k}{2^i}(x)+\es{n-k}{2^i}(\ell)\mid i=0,\dots,r-1)+Q_k=\\
    =(\es{k}{2^i}(x)&+\es{n-k}{2^i}(\ell)\mid 1\leq 2^i\leq\operatorname{rk}(\mathbf{L}))+(\es{k}{2^i}(x)\mid \operatorname{rk}(\mathbf{L})<2^i\leq 2^{r-1})+Q_k.
\end{align*}

Since $\mathbf{L}\in\Fq^{(n-k)\times k}$, we have 
\begin{equation*}
    \operatorname{rk}(\mathbf{L})\leq\min\{k,n-k\}\leq\fl*{\frac{n}{2}}\quad\text{and}\quad{2^{r-1}\leq\fl*{\frac{n}{2}}}< 2^r,
\end{equation*}
so $\operatorname{rk}(\mathbf{L})\leq 2^r-1$. As we are not aware of any relation between $\operatorname{rk}(\mathbf{L})$ and $2^{r-1}$, we rewrite $J_k'$ as
\begin{align*}
    J'_k=(\es{k}{j}(x)+\es{n-k}{j}(\ell)\mid 1\leq j\leq\operatorname{rk}(\mathbf{L}))+(\es{k}{j}(x)\mid \operatorname{rk}(\mathbf{L})<j\leq 2^{r}-1)+Q_k.
\end{align*}
Equality holds, since $\es{k}{j}(x_1,\dots,x_k)+\es{n-k}{j}(x_{k+1},\dots,x_n)\in J_n$ for every $j=1,\dots,2^r-1$ by Lemma~\ref{JnKnK'n}, which implies that $\es{k}{j}(x)+\es{n-k}{j}(\ell)\in J_k'$.\\

We are now ready to estimate the degree of regularity of $J_k'$, hence the complexity of solving the polynomial system $S'=0$ to solve the ESDP. We proceed similarly to what we have done for $d_{\operatorname{reg}}(I_n)$, even if this time we will be able to give just a bound, and not to get the exact value of the degree of regularity. We restrict to the case $n\geq 4$, since the other cases are not cryptographically relevant.

\begin{theo}\label{thm:dregS'}
Suppose that $n\geq 4$. The degree of regularity of the ideal $J'_k$ is bounded by
    \begin{equation*}
        d_{\operatorname{reg}}(J_k')\leq\begin{cases}
            \fl*{\frac{k+\operatorname{rk}(\mathbf{L})}{2}}+2=2^r+1&\text{if }n=2^{r+1}-1,\,k\geq 2^r,\\
            &\text{and }\operatorname{rk}(\mathbf{L})=n-k\text{ (\,$\mathbf{L}$ has maximal rank)},\\
            \\
            \fl*{\frac{k+\operatorname{rk}(\mathbf{L})}{2}}+1&\text{otherwise}.
        \end{cases} 
    \end{equation*}
\end{theo}
\begin{proof}
We divide the proof in three cases, one of which is divided in three subcases.
\begin{itemize}
\item If $k\leq 2^r-1$, then $\es{k}{j}(x)=\es{k}{j}$ vanishes for every $j> k$, so $J'_k=(\es{k}{j}(x)+\es{n-k}{j}(\ell)\mid 1\leq j\leq\operatorname{rk}(\mathbf{L}))+(\es{k}{j}(x)\mid \operatorname{rk}(\mathbf{L})<j\leq k)+Q_k.$
If $\operatorname{rk}(\mathbf{L})<k$, by applying Proposition~\ref{propAlgoESP} with $u=\operatorname{rk}(\mathbf{L})$ to $J_k'$, we obtain
\begin{equation}\label{dregperj}
    d_{\operatorname{reg}}(J_k')\leq \fl*{\frac{k+\operatorname{rk}(\mathbf{L})}{2}}+1.
\end{equation}
If instead $\operatorname{rk}(\mathbf{L})=k$, (\ref{dregperj}) still holds, since $d_{\operatorname{reg}}(J_k')\leq d_{\operatorname{reg}}(Q_k)=k+1$.

\item If $2^r\leq k$ and $\operatorname{rk}(\mathbf{L})\leq2^{r}-2$, using Lemma~\ref{lemEspcascata} as in the proof of Proposition~\ref{propAlgoESP}, but starting from the last polynomials of the list, we obtain:
\begin{equation*}
    \begin{array}{c | c c c c c c c c}
         \text{Step 0} & \es{k}{\operatorname{rk}(\mathbf{L})+1} & \es{k}{\operatorname{rk}(\mathbf{L})+2} & \dots & \es{k}{2^r-2} & \es{k}{2^r-1} & \in J_k'\\
         \text{Step 1} & / & \es{k-1}{\operatorname{rk}(\mathbf{L})+2} & \dots & \es{k-1}{2^r-2} & \es{k-1}{2^r-1} & \in J_k'\\
         \vdots & / & / & \vdots & \vdots & \vdots & \in J_k'.\\
    \end{array}
\end{equation*}
If $\operatorname{rk}(\mathbf{L})+1+s\leq 2^r-1\leq k-s$, that is, for
\begin{equation*}
    s\leq\min\{2^r-2-\operatorname{rk}(\mathbf{L}),k-2^r+1\},
\end{equation*} in Step $s$ we have $\es{k-s}{\operatorname{rk}(\mathbf{L})+1+s},\dots,\es{k-s}{2^r-1}\in J_k'$. We analyze three different cases, depending on the value of $k+\operatorname{rk}(\mathbf{L})$. Notice that $k+\operatorname{rk}(\mathbf{L})\leq n\leq 2^{r+1}-1$.
\begin{itemize}
\item If $k+\operatorname{rk}(\mathbf{L})\leq 2^{r+1}-3$, then the last meaningful step happens for $\bar s=k-2^r+1$. In this case,
$\es{2^r-1}{k+\operatorname{rk}(\mathbf{L})-2^r+2},\dots,\es{2^r-1}{2^r-1}\in J_k'$. Applying Proposition~\ref{propAlgoESP} with $2^r-1$ variables and $u=k+\operatorname{rk}(\mathbf{L})-2^r+1$, we get that
$$d_{\operatorname{reg}}(J_k')\leq\fl*{\frac{k+\operatorname{rk}(\mathbf{L})}{2}}+1.
$$
\item If $k+\operatorname{rk}(\mathbf{L})=2^{r+1}-2$, then the last meaningful step happens for $\bar s=2^r-2-\operatorname{rk}(\mathbf{L})=k-2^r$.
In this case, $\es{2^r}{2^r-1}\in J_k'$, therefore $x_{2^r}\cdot\es{2^r}{2^r-1}=x_{2^r}\cdot\es{2^r-1}{2^r-1}+x_{2^r}^2\cdot\es{2^r-1}{2^r-2}\in J_k'$, hence $\es{2^r}{2^r}\in J'_k$. The same can be done analogously with the other variables, and since $\fl*{\frac{k+\operatorname{rk}(\mathbf{L})}{2}}+1=\fl*{\frac{2^{r+1}-2}{2}}+1=2^r$, bound \eqref{dregperj} holds also in this case, as in the proof of Proposition~\ref{propAlgoESP}.
\item If $k+\operatorname{rk}(\mathbf{L})=2^{r+1}-1$, then the last meaningful step happens for $\bar s=2^r-2-\operatorname{rk}(\mathbf{L})=k-2^r-1$.
In this case, $\es{2^r+1}{2^r-1}\in J_k'$, hence
\begin{align*}
    x_{2^r+1}x_{2^r}&\cdot\es{2^r+1}{2^r-1}
    =x_{2^r+1}x_{2^r}(\es{2^r}{2^r-1}+x_{2^r+1}\es{2^r}{2^r-2})=\\
    &=x_{2^r+1}x_{2^r}(\es{2^r-1}{2^r-1}+x_{2^r}\es{2^r-1}{2^r-2}+x_{2^r+1}\es{2^r}{2^r-2})=\\
    &=x_{2^r+1}x_{2^r}\es{2^r-1}{2^r-1}+x_{2^r+1}x_{2^r}^2\es{2^r-1}{2^r-2}+x_{2^r+1}^2x_{2^r}\es{2^r}{2^r-2}=\\
    &=\es{2^r+1}{2^r+1}+x_{2^r+1}x_{2^r}^2\es{2^r-1}{2^r-2}+x_{2^r+1}^2x_{2^r}\es{2^r}{2^r-2}
\end{align*}
belongs to $J_k'$, and therefore $\es{2^r+1}{2^r+1}\in J_k'$. We obtain the bound
\begin{equation}\label{dregperjpeggio}
    d_{\operatorname{reg}}(J_k')\leq \fl*{\frac{k+\operatorname{rk}(\mathbf{L})}{2}}+2=2^r+1.
\end{equation}
Notice that this case only occurs when $n=2^{r+1}-1$ and $\operatorname{rk}(\mathbf{L})=n-k\leq 2^r-2\leq k-2$, which implies $k\geq 2^r+1$. 
\end{itemize}
\item If $2^r\leq k$ and $\operatorname{rk}(\mathbf{L})\geq 2^{r}-1$, then $n\geq k+\operatorname{rk}(\mathbf{L})\geq 2^{r+1}-1$. This forces $n=2^{r+1}-1$, $k=2^r$ and $\operatorname{rk}(\mathbf{L})=n-k=2^r-1$. Then \eqref{dregperjpeggio} holds, since $d_{\operatorname{reg}}(J_k')\leq k+1=\fl*{\frac{k+\operatorname{rk}(\mathbf{L})}{2}}+2$.
\end{itemize}
\end{proof}

Since the polynomials in $S'$ have degree smaller than or equal to $2^{r-1}$, by combining Theorem~\ref{thm:dregS'} and Theorem~\ref{thm:salizz} we obtain
\begin{cor}\label{cordsolS'}
suppose that $n\geq 4$. The solving degree of the polynomial system $S'=0$ is bounded by
    \begin{equation*}
        d_{\operatorname{sol}}(S')\leq\begin{cases}
            2^r+2 & \text{if }n=2^{r+1}-1,\,k\geq 2^r,\\
            &\text{and }\operatorname{rk}(\mathbf{L})=n-k,\\
            \\
            \max\left\{\fl*{\frac{k+\operatorname{rk}(\mathbf{L})}{2}}+2,2^{r-1}\right\} & \text{otherwise}.\\
        \end{cases}
    \end{equation*}
\end{cor}
While the model $S'$ solves a relaxed version of ESDP, the total complexity of solving the ESDP using the model $S'$ is dominated by the complexity of solving $S'=0$. Moreover, the solving degree of $S'$ is almost always smaller than the solving degree of $S$, which we estimated in Theorem~\ref{cordsolS}. Therefore, it is more efficient to solve the ESDP using the model $S'$. In contrast with the case of $S$, the bound on the solving degree of $S'$ depends on the dimension $k$ of the code and the rank of the matrix $\mathbf{L}$, and therefore on the specific instance of the ESDP. 

In the next section, we briefly discuss how to refine the bounds of Theorem~\ref{thm:dregS'} and Corollary~\ref{cordsolS'} by relating the rank of the matrix $\mathbf{L}$ to invariants of the code $\mathbf{C}$. 

\subsection{Bounds on the rank of the matrix $\mathbf{L}$}\label{sezL}
Given a parity check matrix $\mathbf{H}\in\F_2^{(n-k)\times n}$ of the code $\mathbf{C}$, we consider a matrix $\mathbf{L}\in\F_2^{(n-k)\times k}$ such that $(\mathbf{L}\mid \mathbf{I}_{n-k})$ is the parity check matrix of a code equivalent to $\mathbf{C}$. A matrix $\mathbf{L}$ exists, since $\mathbf{H}$ is full rank, however it is not unique. In fact, different choices of $\mathbf{L}$ may even have different ranks, as the next example shows. 
\begin{ex}
\begin{comment}
\begin{equation*}
    \mathbf{H}=\left(\begin{array}{c c : c c}
        1 & 0 & 1 & 0 \\
        1 & 1 & 0 & 1
    \end{array}\right)
    \overset{C_1\leftrightarrow C_4}{\longleftrightarrow}
    \left(\begin{array}{c c c c}
        0 & 0 & 1 & 1 \\
        1 & 1 & 0 & 1
    \end{array}\right)
    \overset{R_1\leftrightarrow R_1+R_2}{\longleftrightarrow}
    \left(\begin{array}{c c : c c}
        1 & 1 & 1 & 0 \\
        1 & 1 & 0 & 1
    \end{array}\right)=\mathbf{H'},
\end{equation*}
\end{comment}
Let $$\mathbf{H}=\left(\begin{array}{c c : c c}
        1 & 0 & 1 & 0 \\
        1 & 1 & 0 & 1
    \end{array}\right)\quad \mbox{ and }\quad
    \mathbf{H'}=\left(\begin{array}{c c : c c}
        1 & 1 & 1 & 0 \\
        1 & 1 & 0 & 1
    \end{array}\right).$$
It is easy to check that $\mathbf{H}$ and $\mathbf{H'}$ are parity check matrix of two equivalent codes. However, the corresponding $\mathbf{L}$ and $\mathbf{L'}$ have rank $2$ and $1$, respectively.
\end{ex}

By Theorem~\ref{thm:dregS'}, the matrix $\mathbf{L}$ of least rank corresponds to the system $S'$ with predicted smallest solving degree.
It is not always clear, however, how low $\operatorname{rk}(\mathbf{L})$ can be and how one can produce an $\operatorname{rk}(\mathbf{L})$ of small rank, starting from a given parity check matrix $\mathbf{H}$. The rest of the section is devoted to the study of these questions.

We start with a lemma that computes the rank of $\mathbf{L}$ for codes with large minimum distance.
\begin{lem}
    If the minimum distance of the code is $d(\mathbf{C})\geq k$, then any matrix $\mathbf{L}$ as above has full rank.
\end{lem}
\begin{proof}
    Let $a_1,\dots,a_k\in\F_2$ be such that $a_1\ell^c_1+\dots +a_k\ell^c_k=0$, where $\ell^c_i$ denotes the $i$-th column of $\mathbf{L}$, $i=1,\dots,k$. Therefore, the vector $a=(a_1,\dots,a_k,0,\dots,0)\in\F_2^{n}$ is such that $\mathbf{H}\cdot a^\top=(\mathbf{L}\mid \mathbf{I}_{n-k})\cdot a^\top=0$. Since $\mathbf{H}$ is a parity check matrix of $\mathbf{C}$, this implies that $a$ is a codeword of $\mathbf{C}$. Since $\operatorname{wt}(a)\leq k<d(\mathbf{C})$, the only possible case is $a=0$, which means that the columns of $\mathbf{L}$ are linearly independent and therefore $\operatorname{rk}(\mathbf{L})$ is maximum.
    
\end{proof}

This means that for codes with a sufficiently large minimum distance, there is no optimal choice for the matrix $\mathbf{L}$. The situation is different if $d(\mathbf{C})\leq k$. In this case, we consider a generalization of the minimum distance, the \textit{generalized Hamming weights}. In order to define them, we start by defining the \textit{support} of a codeword $c=(c_1,\dots,c_n)\in\mathbf{C}$ as $$\operatorname{supp}(c)=\{i\in[n]\mid c_i\neq0\}.$$ The support of a subcode $\mathbf{D}\leq\mathbf{C}$ is defined as $$\operatorname{supp}(\mathbf{D})=\bigcup\limits_{x\in\mathbf{D}}\operatorname{supp}(x).$$
\begin{defin}[Generalized Hamming weights -- GHWs]
    For $1\leq i\leq k$, the $r$-th \emph{generalized Hamming weight} of an $[n,k]$-linear code $\mathbf{C}$ is
    \begin{equation*}
        d_i(\mathbf{C})=\min\{|\operatorname{supp}(\mathbf{D})| : \mathbf{D}\leq\mathbf{C},\,\operatorname{dim}(\mathbf{D})=i\}.
    \end{equation*}
\end{defin}
Notice that $d_1(\mathbf{C})=d(\mathbf{C})$. For the main properties of GHWs, we refer to~\cite[Chapter 7.10]{HuffmanPless2003Fundamentals}, where one can find the next result and more.
\begin{theo}[Monotonicity and Generalized Singleton bound]
    For an $[n,k]$-linear code $\mathbf{C}$, we have
    \begin{itemize}
        \item \textit{monotonicity}: $1\leq d_1(\mathbf{C})<d_2(\mathbf{C})<\dots<d_k(\mathbf{C})\leq n$;
        \item \textit{generalized Singleton bound}: $d_i(\mathbf{C})\leq n-k+i$.
    \end{itemize}
\end{theo}

If $d(\mathbf{C})\leq k$, one can look for the maximum $i\in[k]$ such that $d_i(\mathbf{C})\leq k$. This idea is related to our problem as follows.
\begin{prop}
    Given a parity check matrix $\mathbf{H}=(\mathbf{L}\mid \mathbf{I}_{n-k})$ of an $[n,k]$-linear code $\mathbf{C}$ such that $d(\mathbf{C})\leq k$, we have 
    \begin{equation*}
        \operatorname{rk}(\mathbf{L})\geq k-\max\{i\in[k]\mid d_i(\mathbf{C})\leq k\}.
    \end{equation*}
\end{prop}
\begin{proof}
    Since $d(\mathbf{C})\leq k$, the set $\{i\in[k]\mid d_i(\mathbf{C})\leq k\}$ is not-empty, and therefore we can consider its maximum $r=\max\{i\in[k]\mid d_i(\mathbf{C})\leq k\}$. If $r=k$, then $\operatorname{rk}(\mathbf{L})\geq 0$ is obviously true. Otherwise, we have that for any $j>r$, $d_j(\mathbf{C})>k$, and therefore for any subcode $\mathbf{D}$ of $\mathbf{C}$ such that $\operatorname{dim}(\mathbf{D})=j>r$, we have $|\operatorname{supp}(\mathbf{D})|>k$. Therefore, any subcode $\mathbf{D}\leq\mathbf{C}$ such that $|\operatorname{supp}(\mathbf{D})|\leq k$, we have $\operatorname{dim}(\mathbf{D})\leq r$.
    
    We now consider the subcode $\mathbf{V}=\mathbf{C}\cap\{(x_1,\dots,x_n)\in\F_2^n\mid x_{k+1}=\ldots=x_n=0\}$. The support of $\mathbf{V}$ is contained in $\{1,\dots,k\}$, and therefore $|\operatorname{supp}(\mathbf{V})|\leq k$, and thus $\operatorname{dim}(\mathbf{V})\leq r$ for the above considerations. It is also true that $\mathbf{V}\cong\operatorname{Ker}(\mathbf{L})$, since for any $x\in\F_2^k$ it holds
    \begin{equation*}
        \mathbf{L}\cdot x^\top=0\iff \mathbf{L}\cdot x^\top+\mathbf{I}_{n-k}\cdot0^\top=0\iff\mathbf{H}\cdot(x\mid 0)^\top=0\iff(x\mid 0)\in\mathbf{V}.
    \end{equation*}
    Therefore, we have $\operatorname{dim}(\operatorname{Ker}(\mathbf{L}))\leq r$, which implies $\operatorname{rk}(\mathbf{L})=\operatorname{dim}(\operatorname{Im}(\mathbf{L}))\geq k-r$ as desired.
    
\end{proof}

We now have a lower bound for $\operatorname{rk}(\mathbf{L})$, which can be used to refine the bound given in Corollary~\ref{cordsolS'}. However, there are two main difficulties with this approach. First, computing the GHWs of a code may be hard. Although there are many theoretical results, there is no general efficient algorithm to compute the GHWs of a generic linear code (see e.g.~\cite{10.1145/3773284}). Since the problem of computing the minimum distance $d_1(\mathbf{C})$ of a code is NP-complete~\cite{1055873}, the problem of computing the GHWs is also NP-complete. Second, it is not clear how to compute the matrix $\mathbf{L}$ of least rank, or even one of relatively small rank. This problem is related to understanding the relations of linear dependence among the columns of $\mathbf{H}$, which is in turn related to the problem of finding the \textit{minimal codewords} of $\mathbf{C}$. The minimal codewords are defined as the codewords of $\mathbf{C}$ whose support does not properly contain the support of any other non-zero codeword. For a general linear code, the problem of finding its minimal codewords is again NP-complete. Nevertheless, the GHWs (or bounds on them) are known for several families of codes. See for example~\cite{133259,135653,651015,LEE2015265, KangquanLi2022AdvancesinMathematicsofCommunications}.

\section{Some variants of the binary ESDP}\label{secVariants}

In this section, we briefly discuss how to apply some of the ideas from the previous section to other contexts. More specifically, we discuss the BSDP, fields of size $q>2$, and the Regular variant.

\subsection{Bounded Syndrome Decoding Problem -- BSDP}\label{secBSDP}
So far we discussed the ESDP, where we search for an error of weight exactly $t$. This is the case of main interest from the cryptographic point of view. Intuitively, in cryptography we prefer to add as much noise as possible to the plaintext, therefore we often add an error of weight $t$ equal to the maximum error-correction capability of the code. In coding theory, one is more likely to be interested in the BSDP, as one wants to correct every error whose weight is smaller than or equal to the maximum error-correction capability of the code. The two variants of the SDP, however, are strictly related. In fact, one can solve the ESDP for a target weight $t$ by solving the BSDP for $t-1$ and $t$, while solving the ESDP for $1,\dots,t$ allows to solve the BSDP for $t$. One expects that the complexity of the ESDP grows with $t$ and that the complexity of the BSDP for $t$ is dominated by the last step, i.e. by the ESDP for $t$. If this is the case, then the ESDP and the BSDP for the same $t$ have the same asymptotic complexity.

We now briefly describe how to apply the ideas from the previous section to the BSDP. The core is the following lemma, which is a variation of Lemma~\ref{lemESPweight}:
\begin{lem}\label{lemBSDP2}
    For any target weight $0\leq t\leq n$,
    \begin{equation*}
        \{x\in\F_2^n\mid\es{n}{j}(x)=0\text{ for all }j=t+1,\dots,n\}=\{x\in\F_2^n\mid\operatorname{wt}(x)\leq t\}.
    \end{equation*}
\end{lem}
\begin{proof}
    "$\subseteq$": suppose that $\operatorname{wt}(x)>t$. Then, by Lemma~\ref{lemESPweight}, we have that $\es{n}{\operatorname{wt}(x)}(x)=\binom{\operatorname{wt}(x)}{\operatorname{wt}(x)}=1$, which is a contradiction. Therefore, $\operatorname{wt}(x)\leq t$;
    
    "$\supseteq$": for any $j=t+1,\dots,n$, by Lemma~\ref{lemESPweight} we have $\es{n}{j}(x)=\binom{\operatorname{wt}(x)}{j}$, which is equal to $0$ since $j>t\geq\operatorname{wt}(x)$.
    
\end{proof}
This allows us to use the elementary symmetric polynomials to efficiently encode the weight constraint in the BSDP. Let $S_{3B}=\{\es{n}{j}(x)\mid j=t+1,\dots,n\}$. The resulting polynomial system $S_B=S_1\cup S_2\cup S_{3B}$ is such that the solutions $S_B=0$ are exactly the solutions of the BSDP.

Notice that, if $x\in\F_2^n$ is such that $\operatorname{wt}(x)=t<2^r$, where $r=\fl*{\operatorname{log}_2 n}$, then, due to \eqref{eq:BuildingBlocksESP}, we have $\es{n}{j}(x)=0$ for any $j\geq 2^r$, since $t_r=0$ in the binary decomposition of $t$. Therefore, if the target weight is $t<2^r$, it suffices to let $S_{3B}=\{\es{n}{j}(x)\mid j=t+1,\dots,2^r\}$, since the equations corresponding to higher degrees are also satisfied. This has the effect of reducing the degree of the polynomials in the system, similarly to the model $S'$ for the ESDP.

Concerning the complexity of this model, the same considerations as in Section~\ref{secFirstModel} lead to focus on the degree of regularity of the symmetric ideal $(\es{n}{t+1},\dots,\es{n}{n})+Q_n$. Applying Proposition~\ref{propAlgoESP} with $u=t$, we obtain
\begin{equation*}
    d_{\operatorname{reg}}(S_B)\leq\fl*{\frac{n+t}{2}}+1.
\end{equation*}
The bound $d_{\operatorname{reg}}(S_B)\leq k+1$ holds also in this case, since this holds for every system that contains $n-k$ linear equations and the field equations. A bound on the solving degree can be derived from Theorem~\ref{thm:salizz} as usual.

The idea of dropping the equation of largest degree, instead, may not lead to similar results, it introduces too many possible weights among the solutions of the polynomial system. Therefore, with the variable substitution we may lower the degree of regularity or get a dependence with the specific instance of the BSDP, but we will not get a better bound for the solving degree.

\subsection{Field of size $q>2$}\label{secSizeq>2}
If we consider a field $\F_q$ of size $q>2$, the proposed model for the ESDP does not have the same efficiency, as it relies on Lemma~\ref{lemESPweight}, which is not true in the general case of $q\geq 2$. One may solve several polynomial systems, each one corresponding to a specific combination of $t$ non-zero coefficients in $\F_q$, where $t$ is the target weight, but it would soon be computationally unfeasible. 

Another way may be, as in~\cite{caminata2025quadraticmodelingssyndromedecoding}, to add some new variables $\{z_i\}_{i=1}^n$ which represent the non-zero positions of the vector. This was done by adding the set of polynomials $\{x_i^{q-1}-z_i\mid i=1,\dots,n\}$ to their model. In combination with field equations $\{x_i^q-x_i\mid i=1,\dots,n\}$, they force $z_i$ to be $0$ if $x_i=0$, and $1$ otherwise. In our case, we can slightly modify our model for the ESDP over $\F_2$, by adding this set of polynomials to the usual parity check encoding, the finite field equations and the weight constraint encoding, where the last one is expressed by $\{\es{n}{2^i}(z)-t_i\mid i=0,\dots,r\}$, $r=\fl*{\operatorname{log}_2 n}$. Since we can also add the equations $\{z_i^2-z_i\mid i=1,\dots,n\}$, the degree of regularity of the given system is the degree of regularity of the ideal $I_n^q=(x_i^{q-1},z_i^2\mid i=1,\dots,n)+(\es{n}{2^i}(z)\mid i=0,\dots,r)+(\sum\limits_{j=1}^n h_{i,j}x_j\mid i=1,\dots,n-k)$. As in the proof of Proposition~\ref{propAlgoESP}, one can show that every $z$-variable permutation of $z_1\cdots z_{\fl*{\frac{n}{2}}+1}$ belongs to this ideal. Therefore, in the quotient $\F_q[x_1,\dots,x_n,z_1,\dots,z_n]/I_n^q$, a non-zero monomial of highest degree must divide a $z$-variable permutation of $x_1^{q-2}\cdots x_n^{q-2}\cdot z_1\cdots z_{\fl*{\frac{n}{2}}}$. After the $x$-variable substitution given by the parity check encoding, we have that the degree of regularity is bounded from above by $k(q-2)+\fl*{\frac{n}{2}}+1$. As $q$ increases, this bound becomes significantly larger than the one obtained in the binary case.\\

For the BSDP in the case $q>2$, %we can make similar considerations, but this time we do not even need the additional variables $z_i$, since Lemma~\ref{lemBSDP2} also holds when we consider vectors in $\F_q^n$. Therefore, 
one may use the same model proposed in Section~\ref{secBSDP}, with the proper field equations. %In order to bound the degree of regularity associated with the system, consider the usual related quotient. After the substitution of variables given by the parity check encoding, the highest possible degree among the monomials is reached by $x_1^{q-1}\cdots x_k^{q-1}$ and its conjugates under the action of the symmetric group. 
The degree of regularity of any system that contains the field equations of $\mathbb{F}_q$ is bounded from above by $k(q-1)+1$. This yields a simple, although inaccurate, bound for the solving degree of the system.

\subsection{Regular Syndrome Decoding Problem}\label{sezRSDP}
The Regular variant of the SDP, first introduced in~\cite{10.1007/11554868_6}, is defined, in the usual setting, as follows:  
\begin{prob}[RSDP: Regular Syndrome Decoding Problem]
    Find a vector $e\in\F_q^n$ such that $\mathbf{H}\cdot e^\top=s^\top$ and $e=(e_1,\dots,e_t)$, where $t\mid n$, $e_\ell\in\Fq^{\frac{n}{t}}$ and $\operatorname{wt}(e_\ell)=1$ for $\ell=1,\dots,t$.
\end{prob}
In this variant, therefore, the solution vector $e$ is divided into $t$ equal-length blocks, each containing exactly one non-zero element. The model described in Section~\ref{secCore} can be applied as follows. If we divide the variables into $t$ sets $X_1=\{x_1,\dots,x_h\},\dots,X_t=\{x_{n-h+1},\dots,x_n\}$, where $h=\frac{n}{t}$, we can consider the weight constraint encoding over each set. For example, in the binary case, for any $\ell=1,\dots,t$ we can impose the condition $\operatorname{wt}(e_\ell)=1$ by considering the set $S_{3_\ell}=\{\es{h}{1}(X_\ell)+1\}\cup\{\es{h}{2^i}(X_\ell)\mid i=1,\dots,\fl*{\operatorname{log}_2 h}\}$. Therefore, the set $S_R=S_1\cup S_2\cup\bigcup\limits_{\ell=1}^t S_{3_\ell}$ is such that the solutions of $S_R=0$ are the solutions of the RSDP instance. 

To study the complexity of this variant model, we again consider the corresponding degree of regularity. Thanks to Lemma~\ref{lem:belprodottoesp}, we can prove that for any $\ell=1,\dots,t$ the ideal generated by $S_R^{\operatorname{top}}$ contains all the elementary symmetric polynomials $\es{h}{j}(X_\ell)$, $j=1,\dots,h$, and by Lemma~\ref{full prod} it is possible to recursively construct the elementary symmetric polynomials in an ever-increasing number of variables, until reaching the polynomials in the full set of $n$ variables. We can therefore apply Proposition~\ref{propAlgoESP} and obtain
\begin{equation*}
    d_{\operatorname{reg}}(S_R)\leq\fl*{\frac{n}{2}}+1.
\end{equation*}
This time, the degrees of the system polynomials are lower than the degree of regularity, since $h\leq\frac{n}{2}$, therefore Theorem~\ref{thm:salizz} gives us directly that
\begin{equation*}
    d_{\operatorname{sol}}(S_R)\leq\fl*{\frac{n}{2}}+2.
\end{equation*}
Although this result is less favourable than the one for the model for the RSDP proposed by Briaud and \O{}ygarden in~\cite{10.1007/978-3-031-30589-4_14}, we notice that our model could more easily be adapted to a \textit{less regular} version of the RSDP, such as, for example, to the case of blocks of different lengths or with different weights.

\printbibliography

\end{document}